\documentclass[11pt]{article}
\usepackage[T1]{fontenc}
\usepackage{lmodern}

\iffalse
\usepackage[accepted]{icml2021_arxiv}
\else
\usepackage{fullpage}
\usepackage{algorithm}
\usepackage{algorithmic}
\usepackage{natbib}
\setcitestyle{authoryear,round,citesep={;},aysep={,},yysep={;}}

\renewcommand{\cite}[1]{\citep{#1}}
\fi

\usepackage{xcolor}
\usepackage{amsmath}
\usepackage{amssymb}
\usepackage{amsthm}
\usepackage[]{algorithm}
\usepackage{cleveref}
\usepackage{graphicx}
\usepackage{subcaption}
\usepackage[font={small}]{caption}
\usepackage{thm-restate}
\usepackage[normalem]{ulem}
\usepackage{float}
\usepackage{balance}

\usepackage{enumitem}
\usepackage{xcolor}
\newcommand{\nc}{\newcommand}
\newcommand{\DMO}{\DeclareMathOperator}
\DeclareMathAlphabet\mathbfcal{OMS}{cmsy}{b}{n}

\newcommand{\badih}[1]{\textcolor{red}{[Badih: #1]}}

\newcommand{\silly}{\hspace*{1em}}

\newcommand{\Procedure}[2]{\STATE {{\bf procedure} {\sc #1}}} 
\newcommand{\State}{\STATE \silly}
\newcommand{\SubState}{\STATE \silly \silly}
\newcommand{\SubSubState}{\STATE \silly \silly \silly}
\newcommand{\Return}{{\bf return \/}}
\newcommand{\For}[1]{\State {\bf for\/} #1}
\newcommand{\SubFor}[1]{\SubState {\bf for\/} #1}
\newcommand{\If}[1]{\State {\bf if\/} #1}

\newcommand{\EndProcedure}{}
\newcommand{\EndIf}{}
\newcommand{\EndFor}{}

\allowdisplaybreaks

\newcommand{\poly}{\mathrm{poly}}
\nc{\MS}{\mathcal{S}}
\nc{\MP}{\mathcal{P}}
\nc{\MR}{\mathcal{R}}
\nc{\cM}{\mathcal{M}}
\nc{\cS}{\mathcal{S}}
\nc{\cI}{\mathcal{I}}
\nc{\cA}{\mathcal{A}}
\nc{\tcA}{\tilde{\cA}}

\nc{\MZ}{\mathcal{Z}}

\DMO{\Binom}{Binom}
\newcommand{\E}{\mathbb{E}}
\DMO{\Var}{Var}

\newcommand{\bb}{\mathbf{b}}
\newcommand{\bh}{\mathbf{h}}
\newcommand{\tbh}{\tilde{\bh}}
\newcommand{\bB}{\mathbf{B}}
\newcommand{\bc}{\mathbf{c}}
\newcommand{\bC}{\mathbf{C}}
\newcommand{\bX}{\mathbf{X}}
\nc{\tbx}{\tilde{\bx}}
\nc{\tbX}{\tilde{\bX}}
\nc{\tZ}{\tilde{Z}}
\nc{\tz}{\tilde{z}}
\newcommand{\bU}{\mathbf{U}}
\nc{\tbU}{\tilde{\bU}}
\newcommand{\bT}{\mathbf{T}}
\nc{\tbT}{\tilde{\bT}}
\newcommand{\bD}{\mathbf{D}}
\nc{\tbD}{\tilde{\bD}}

\newcommand{\bu}{\mathbf{u}}
\newcommand{\bw}{\mathbf{w}}
\newcommand{\bt}{\mathbf{t}}
\newcommand{\bx}{\mathbf{x}}
\newcommand{\bbf}{\mathbf{f}}
\newcommand{\by}{\mathbf{y}}

\newcommand{\br}{\mathbf{r}}

\newcommand{\bv}{\mathbf{v}}
\newcommand{\bz}{\mathbf{z}}
\newcommand{\bq}{\mathbf{q}}

\newcommand{\R}{\mathbb{R}}
\newcommand{\N}{\mathbb{N}}

\newcommand{\hr}{\hat{r}}
\newcommand{\hp}{\hat{p}}

\nc{\BN}{\mathbb{N}}
\newcommand{\Z}{\mathbb{Z}}
\nc{\BZ}{\mathbb{Z}}

\newcommand{\cL}{\mathcal{L}}
\newcommand{\cE}{\mathcal{E}}

\newcommand{\bone}{\mathbf{1}}
\newcommand{\bzero}{\mathbf{0}}

\newcommand{\bA}{\mathbf{A}}
\nc{\tbA}{\tilde{\bA}}
\newcommand{\bI}{\mathbf{I}}
\newcommand{\bQ}{\mathbf{Q}}
\newcommand{\bs}{\mathbf{s}}

\newcommand{\cD}{\mathcal{D}}
\newcommand{\tcD}{\tilde{\cD}}
\newcommand{\hcD}{\hat{\cD}}
\newcommand{\cbD}{\mathbfcal{D}}
\newcommand{\cDcentral}{\cD^{\mathrm{central}}}

\DeclareMathOperator{\supp}{supp}

\DeclareMathOperator{\DLap}{DLap}

\DeclareMathOperator{\NB}{NB}

\newcommand{\pad}{\mathrm{ext}}
\newcommand{\mysim}{\mathrm{sim}}
\newcommand{\mysum}{\mathrm{sum}}
\newcommand{\cor}{\mathrm{cor}}

\newcommand{\centralalg}{\textsc{CorrNoise}}
\newcommand{\centralalgm}{\textsc{CorrNoiseCentral}}
\newcommand{\randomizer}{\textsc{CorrNoiseRandomizer}}
\newcommand{\analyzer}{\textsc{CorrNoiseAnalyzer}}
\newcommand{\sparsevectorrandomizer}{\textsc{SparseVectorCorrNoiseRandomizer}}
\newcommand{\sparsevectoranalyzer}{\textsc{SparseVectorCorrNoiseAnalyzer}}

\newcommand{\eps}{\ensuremath{\varepsilon}}
\renewcommand{\epsilon}{\eps}

\newcommand{\bepsilon}{\boldsymbol{\epsilon}}

\newtheorem{theorem}{Theorem}

\newtheorem{observation}[theorem]{Observation}
\newtheorem{lemma}[theorem]{Lemma}

\newtheorem{definition}[theorem]{Definition}

\newtheorem{corollary}[theorem]{Corollary}

\begin{document}

\allowdisplaybreaks
\setlength{\abovedisplayskip}{2pt}
\setlength{\belowdisplayskip}{2pt}

\onecolumn 

\iffalse
\icmltitle{Differentially Private Aggregation in the Shuffle Model: \\ Almost Central Accuracy in Almost a Single Message}

\begin{icmlauthorlist}
\icmlauthor{Badih Ghazi}{googlemtv}
\icmlauthor{Ravi Kumar}{googlemtv}
\icmlauthor{Pasin Manurangsi}{googlemtv}
\icmlauthor{Rasmus Pagh}{ucopenhagen,googlemtv}
\icmlauthor{Amer Sinha}{googlesanbruno}
\end{icmlauthorlist}

\icmlaffiliation{ucopenhagen}{University of Copenhagen, Denmark}
\icmlaffiliation{googlemtv}{Google Research, Mountain View}
\icmlaffiliation{googlesanbruno}{Google, San Bruno}

\icmlcorrespondingauthor{Badih Ghazi}{badihghazi@gmail.com}
\icmlcorrespondingauthor{Ravi Kumar}{ravi.k53@gmail.com}
\icmlcorrespondingauthor{Pasin Manurangsi}{pasin@google.com}
\icmlcorrespondingauthor{Rasmus Pagh}{pagh@di.ku.dk}
\icmlcorrespondingauthor{Amer Sinha}{amersinha@google.com}

\printAffiliationsAndNotice{\icmlEqualContribution}

\vskip 0.3in

\else

\title{Differentially Private Aggregation in the Shuffle Model: \\ Almost Central Accuracy in Almost a Single Message}
\author{Badih Ghazi\footnotemark[1]
\hspace*{0.7cm}
Ravi Kumar\footnotemark[1]
\hspace*{0.7cm}
Pasin Manurangsi\footnotemark[1]
\hspace*{0.7cm}
Rasmus Pagh\footnotemark[2]
\hspace*{0.7cm}
Amer Sinha\footnotemark[3]
\\
\footnotemark[1] Google Research, Mountain View, CA \\
\footnotemark[2] University of Copenhagen, Denmark \\
\footnotemark[3] Google, San Bruno \\
\texttt{\{badihghazi, ravi.k53\}@gmail.com, \{pasin, amersinha\}@google.com, pagh@di.ku.dk}
}
\date{}

\maketitle
\fi

\begin{abstract}
The shuffle model of differential privacy has attracted attention in the literature due to it being a middle ground between the well-studied central and local models. In this work, we study the problem of summing (aggregating) real numbers or integers, a basic primitive in numerous machine learning tasks, in the shuffle model. We give a protocol achieving error arbitrarily close to that of the (Discrete) Laplace mechanism in the central model, while each user only sends $1 + o(1)$ short messages in expectation.
\end{abstract}


\section{Introduction}\label{sec:intro}

A principal goal within trustworthy machine learning is the design of privacy-preserving algorithms. In recent years, differential privacy (DP) \cite{dwork2006calibrating,dwork2006our} has gained significant popularity as a privacy notion due to the strong protections that it ensures. This has led to several practical deployments including by Google~\cite{erlingsson2014rappor,CNET2014Google}, Apple~\cite{greenberg2016apple,dp2017learning}, Microsoft~\cite{ding2017collecting}, and the U.S.~Census Bureau \cite{abowd2018us}. 
DP properties are often expressed in terms of parameters $\epsilon$ and $\delta$, with small values indicating that the algorithm is less likely to leak information about any individual within a set of $n$ people providing data. It is common to set $\epsilon$ to a small positive constant (e.g., $1$), and $\delta$ to inverse-polynomial in $n$.

DP can be enforced for any statistical or machine learning task, and it is particularly well-studied for the 
\emph{real summation} problem, where each user $i$ holds a real number $x_i \in [0, 1]$, and the goal is to estimate $\sum_i x_i$. This constitutes a basic building block within machine learning, with extensions including (private) distributed mean estimation~\citep[see, e.g.,][]{BiswasD0U20, girgis2020shuffled}, stochastic gradient descent~ \cite{song2013stochastic, bassily2014private, abadi2016deep, agarwal2018cpsgd}, and clustering \cite{StemmerK18,Stemmer20}.

The real summation problem, which is the focus of this work, has been well-studied in several models of DP. In the \emph{central} model where a curator has access to the raw data and is required to produce a private data release, the smallest possible absolute error is known to be $O(1/\epsilon)$; this can be achieved via the ubiquitous Laplace mechanism~\cite{dwork2006calibrating}, which is also known to be nearly optimal\footnote{Please see the supplementary material for more discussion.} for the most interesting regime of $\eps \leq 1$. In contrast, for the more privacy-stringent \emph{local} setting \cite{kasiviswanathan2008what}~\citep[also][]{warner1965randomized} where each message sent by a user is supposed to be private, the smallest error is known to be $\Theta_{\eps}(\sqrt{n})$ \cite{beimel2008distributed,ChanSS12}. This significant gap between the achievable central and local utilities has motivated the study of intermediate models of DP. The \emph{shuffle model} \cite{bittau17,erlingsson2019amplification,CheuSUZZ19} reflects the setting where the user reports are randomly permuted before being passed to the analyzer; the output of the shuffler is required to be private. Two variants of the shuffle model have been studied: in the \emph{multi-message} case~\citep[e.g.,][]{CheuSUZZ19}, each user can send multiple messages to the shuffler; in the \emph{single-message} setting each user sends one message~\citep[e.g.,][]{erlingsson2019amplification, BalleBGN19}.

For the real summation problem, it is known that the smallest possible absolute error in the single-message shuffle model\footnote{Here $\tilde{\Theta}_\eps(\cdot)$ hides a polylogarithmic factor in $1/\delta$, in addition to a dependency on $\eps$.} is $\tilde{\Theta}_{\eps}(n^{1/6})$ \cite{BalleBGN19}. In contrast, multi-message shuffle protocols exist with a near-central accuracy of $O(1/\epsilon)$ \cite{ghazi2019private, balle_merged}, but they suffer several drawbacks in that the number of messages sent per user is required to be at least $3$, each message has to be substantially longer than in the non-private case, and in particular, the number of bits of communication per user has to grow with $\log(1/\delta)/\log n$.
This (at least) three-fold communication blow-up relative to a non-private setting can be a limitation in real-time reporting use cases (where encryption of each message may be required and the associated cost can become dominant) and in federated learning settings (where great effort is undertaken to compress the gradients). Our work 
shows that near-central accuracy \emph{and} near-zero communication overhead are possible for real aggregation over sufficiently many users:


\begin{theorem} \label{thm:real-main}
For any $0 < \eps \leq O(1), \zeta, \delta \in (0, 1/2)$, there is an $(\eps, \delta)$-DP real summation  protocol in the shuffle model whose mean squared error (MSE) is at most the MSE of the Laplace mechanism with parameter $(1 - \zeta)\eps$, each user sends 
$1 + \tilde{O}_{\zeta, \epsilon}\left(\frac{\log(1/\delta)}{\sqrt{n}}\right)$
messages in expectation, and each message contains $\tfrac{1}{2} \log n + O(\log \frac{1}{\zeta} )$ bits.
\end{theorem}
Note that $\tilde{O}_{\zeta, \epsilon}$ hides a small $\poly(\log n, 1/\eps, 1/\zeta)$ term. 
Moreover, the number of bits per message is equal, up to lower order terms, to that needed to achieve MSE $O(1)$ even \emph{without} any privacy constraints.

Theorem~\ref{thm:real-main} follows from an analogous result for the case of \emph{integer} aggregation, where each user is given an element in the set $\{0,1,..., \Delta\}$ (with $\Delta$ an integer), and the goal of the analyzer is to estimate the sum of the users' inputs. We refer to this task as the \emph{$\Delta$-summation} problem.

For $\Delta$-summation, the standard mechanism in the central model is the Discrete Laplace (aka Geometric) mechanism, which first computes the true answer and then adds to it a noise term sampled from the Discrete Laplace distribution\footnote{The Discrete Laplace distribution with parameter $s$, denoted by $\DLap(s)$, has probability mass $\frac{1 - e^{-s}}{1 + e^{-s}} \cdot e^{-s|k|}$ at each $k \in \Z$.} with parameter $\eps/\Delta$~\cite{GhoshRS12}. 
We can achieve an error arbitrarily close to this mechanism in the shuffle model, with minimal communication overhead:

\begin{theorem} \label{thm:dsum-main}
For any $0 < \eps \leq O(1), \gamma, \delta \in (0, 1/2), \Delta \in \N$, there is an $(\eps, \delta)$-DP $\Delta$-summation protocol in the shuffle model whose MSE is at most that of the Discrete Laplace mechanism with parameter $(1 - \gamma)\eps/\Delta$, and where each user sends $1 + \tilde{O}\left(\frac{\Delta \log(1/\delta)}{\gamma \epsilon n}\right)$ 
messages in expectation, with each message containing $\lceil \log \Delta \rceil + 1$ bits.
\end{theorem}

In~\Cref{thm:dsum-main}, the $\tilde{O}(\cdot)$ hides a $\poly\log\Delta$ factor.
We also note that the number of bits per message in the protocol is within a single bit from the minimum message length needed to compute the sum without any privacy constraints.  Incidentally, for $\Delta = 1$, \Cref{thm:dsum-main} improves the
communication overhead obtained by~\citet{GKMP20-icml} from $O\left(\frac{\log^2(1/\delta)}{\eps^2 n}\right)$ to $O\left(\frac{\log(1/\delta)}{\eps n}\right)$.  (This improvement turns out to be crucial in practice, as our experiments show.)

Using Theorem~\ref{thm:real-main} as a black-box, we obtain the following corollary for the \emph{$1$-sparse vector summation} problem, where each user is given a $1$-sparse (possibly high-dimensional) vector of norm at most $1$, and the goal is to compute the sum of all user vectors with minimal $\ell_{2}$ error.

\begin{corollary} \label{cor:1-sparse}
For every $d \in \mathbb{N}$, and $0 < \eps \leq O(1), \zeta, \delta \in (0, 1/2)$, there is an $(\eps, \delta)$-DP algorithm for $1$-sparse vector summation in $d$ dimensions in the shuffle model whose $\ell_2$ error is at most that of the Laplace mechanism with parameter $(1 - \zeta)\eps/2$, and where each user sends 
$1 + \tilde{O}_{\zeta, \epsilon}\left( \frac{d \log(1/\delta)}{\sqrt{n}}\right)$
messages in expectation, and each message contains $\log{d} + \tfrac{1}{2} \log n + O(\log \frac{1}{\zeta})$ bits.
\end{corollary}

\subsection{Technical Overview}

We will now describe the high-level technical ideas underlying our protocol and its analysis. Since the real summation protocol can be obtained from the $\Delta$-summation protocol using known randomized discretization techniques (e.g.,~from \citet{balle_merged}), we focus only on the latter. For simplicity of presentation, we will sometimes be informal here; everything will be formalized later.


\paragraph{Infinite Divisibility.} To achieve a similar performance to the central-DP Discrete Laplace mechanism (described before \Cref{thm:dsum-main}) in the shuffle model, we face several obstacles. To begin with, the noise has to be divided among all users, instead of being added centrally. Fortunately, this can be solved through the \emph{infinite divisibility} of Discrete Laplace distributions\footnote{See~\cite{GoryczkaX17} for a discussion on distributed noise generation via infinite divisibility.}: there is a distribution $\cD'$ for which, if each user $i$ samples a noise $z_i$ independently from $\cD'$, then $z_1 + \cdots + z_n$ has the same distribution as $\DLap(\eps/\Delta)$. 

To implement the above idea in the shuffle model, each user has to be able to send their noise $z_i$ to the shuffler. Following~\citet{GKMP20-icml}, we can send such a noise in unary\footnote{Since the distribution $\cD'$ has a small tail probability, $z_i$ will mostly be in $\{0, -1, +1\}$, meaning that non-unary encoding of the noise does not significantly reduce the communication.}, i.e., if $z_i > 0$ we send the $+1$ message $z_i$ times and otherwise we send the $-1$ message $-z_i$ times. This is in addition to user $i$ sending their own input $x_i$ (in binary\footnote{If we were to send $x_i$ in unary similar to the noise, it would require possibly as many as $\Delta$ messages, which is undesirable for us since we later pick $\Delta$ to be $O_{\eps, \delta}(\sqrt{n})$ for real summation.}, as a single message) if it is non-zero. The analyzer is simple: sum up all the messages.

Unfortunately, this zero-sum noise approach is \emph{not} shuffle DP for $\Delta > 1$ because, even after shuffling, the analyzer can still see $u_j$, the number of messages $j$, which is exactly the number of users whose input is equal to $j$ for $j \in \{2, \dots, \Delta\}$. 

\paragraph{Zero-Sum Noise over Non-Binary Alphabets.}
To overcome this issue, we have to 
`noise'' the values $u_j$ themselves, while at the same time preserving the accuracy. We achieve this by making some users send additional messages whose sum is equal to zero; e.g., a user may send $-1, -1, +2$ in conjunction with previously described messages. Since the analyzer just sums up all the messages, this additional zero-sum noise still does not affect accuracy.

The bulk of our technical work is in the privacy proof of such a protocol. To understand the challenge, notice that the analyzer still sees the $u_j$'s, which are now highly \emph{correlated} due to the zero-sum noise added. This is unlike most DP algorithms in the literature where noise terms are added independently to each coordinate. Our main technical insight is that, by a careful change of basis, we can ``reduce'' the view to the independent-noise case.

To illustrate our technique, let us consider the case where $\Delta = 2$. In this case, there are two zero-sum ``noise atoms'' that a user might send: $(-1, +1)$ and $(-1, -1, +2)$. These two kinds of noise are sent independently, i.e., whether the user sends $(-1, +1)$ does not affect whether $(-1, -1, +2)$ is also sent. After shuffling, the analyzer sees $(u_{-1}, u_{+1}, u_{+2})$. Observe that there is a one-to-one mapping between this and $(v_1, v_2, v_3)$ defined by $v_1 := u_{-1} - 2 \cdot u_{+2}, v_2 := u_{-1} - u_{+1} - u_{+2}, v_3 := - u_{-1} + u_{+1} + 2  \cdot u_{+2}$, meaning that we may prove the privacy of the latter instead. Consider the effect of sending the $(+1, -1)$ noise: $v_1$ is increased by one, whereas $v_2, v_3$ are completely unaffected. Similarly, when we send $(-1, -1, +2)$ noise, $v_2$ is increased by one, whereas $v_1, v_3$ are completely unaffected. Hence, the noise added to $v_1, v_2$ are now independent!  Finally, $v_3$ is exactly the sum of all messages, which was noised by the $\DLap$ noise explained earlier.

A vital detail omitted in the previous discussion is that the $\DLap$ noise, which affects $u_{-1}, u_{+1}$, is \emph{not} canceled out in $v_1, v_2$. Indeed, in our formal proof we need a special argument (\Cref{lem:correlated-coord}) to deal with this noise.

Moreover, generalizing this approach to larger values of $\Delta$ requires overcoming additional challenges: (i) the basis change has to be carried out over the \emph{integers}, which precludes a direct use of classic tools from linear algebra such as the Gram--Schmidt process, and (ii) special care has to be taken when selecting the new basis so as to ensure that the sensitivity does not significantly increase, which would require more added noise (this complication leads to the usage of the \emph{$\bQ$-linear query} problem in Section~\ref{sec:privacy_pf_and_params}).




\subsection{Related Work}



\paragraph{Summation in the Shuffle Model.} 
Our work is most closely related to that of~\citet{GKMP20-icml} who gave a protocol for the case where $\Delta = 1$ (i.e., binary summation) and our protocol can be viewed as a generalization of theirs. As explained above, this requires significant novel technical and conceptual ideas; for example, the basis change was not (directly) required by~\citet{GKMP20-icml}. 

The idea of splitting the input into multiple additive shares dates back to the ``split-and-mix'' protocol of~\citet{ishai2006cryptography} whose analysis was improved in~\citet{ghazi2019private,balle_merged} to get the aforementioned shuffle DP algorithms for aggregation. These analyses all crucially rely on the addition being over a finite group.
Since we actually want to sum over integers and there are $n$ users, this approach requires the group size to be at least $n \Delta$ to prevent an ``overflow''.
This also means that each user needs to send at least $\log (n\Delta)$ bits. On the other hand, by dealing with integers directly, each of our messages is only $\lceil \log \Delta \rceil + 1$ bits, further reducing the communication.


From a technical standpoint, our approach is also different from that of~\citet{ishai2006cryptography} as we analyze the privacy of the protocol, instead of its security as in their paper. This allows us to overcome the known lower bound of $\Omega(\frac{\log(1/\delta)}{\log{n}})$ on the number of messages for information-theoretic security~\cite{ghazi2019private}, and obtain a DP protocol with $\tilde{O}_{\zeta, \epsilon}\left(\frac{\log(1/\delta)}{\sqrt{n}}\right)$ messages (where $\zeta$ is as in Theorem~\ref{thm:real-main}).

\paragraph{The Shuffle DP Model.}
Recent research on the shuffle model of DP includes work on aggregation mentioned above \cite{BalleBGN19, ghazi2019private, balle_merged}, analytics tasks including computing histograms and heavy hitters \cite{anon-power, balcer2019separating, ghazi2020pure, GKMP20-icml,Cheu-Zhilyaev-histogramfake}, counting distinct elements \cite{balcer2021connecting,chen2020distributed} and private mean estimation \cite{girgis2020shuffled}, as well as $k$-means clustering \cite{one-round-clustering}.

\paragraph{Aggregation in Machine Learning.}
We note that communication-efficient private aggregation is a core primitive in \emph{federated learning} (see Section $4$ of \citet{kairouz2019advances} and the references therein). It is also naturally related to mean estimation in distributed models of DP~\citep[e.g.,][]{gaboardi2019locally}. Finally, we point out that communication efficiency is a common requirement in distributed learning and optimization, and substantial effort is spent on compression of the messages sent by users, through multiple methods including hashing, pruning, and quantization~\citep[see, e.g.,][]{zhang2013information, alistarh2016qsgd, suresh2017distributed, AcharyaSZ19, chen2020breaking}.

\subsection{Organization}
We start with some background in Section~\ref{sec:prelims}. Our protocol is presented in Section~\ref{sec:protocol}. Its privacy property is established and the parameters are set in Section~\ref{sec:privacy_pf_and_params}. Experimental results are given in Section~\ref{sec:experiments}. We discuss some interesting future directions in Section~\ref{sec:conclusion}. All missing proofs can be found in the Supplementary Material (SM).

\section{Preliminaries and Notation}\label{sec:prelims}

We use $[m]$ to denote $\{1, \dots, m\}$.


\paragraph{Probability.}
For any distribution $\cD$, we write $z \sim \cD$ to denote a random variable $z$ that is distributed as $\cD$. For two distributions $\cD_1, \cD_2$, let $\cD_1 + \cD_2$ (resp., $\cD_1 - \cD_2$) denote the distribution of $z_1 + z_2$ (resp., $z_1 - z_2$) where $z_1 \sim \cD_1, z_2 \sim \cD_2$ are independent. For $k \in \R$, we use $k + \cD$ to denote the distribution of $k + z$ where $z \sim \cD$.

A distribution $\cD$ over non-negative integers is said to be \emph{infinitely divisible} if and only if, for every $n \in \N$, there exists a distribution $\cD_{/n}$ such that $\cD_{/n} + \cdots + \cD_{/n}$ is identical to $\cD$, where the sum is over $n$ distributions.

The negative binomial distribution with parameters $r > 0$, $p \in [0, 1]$, denoted $\NB(r, p)$, has probability mass $\binom{k + r - 1}{k}(1 - p)^r p^k$ at all $k \in \Z_{\geq 0}$. $\NB(r, p)$ is infinitely divisible; specifically, $\NB(r, p)_{/n} = \NB(r/n, p)$.

\paragraph{Differential Privacy.}
Two input datasets $X = (x_1, \dots, x_n)$ and $X' = (x'_1, \dots, x'_n)$ are said to be \emph{neighboring} if and only if they differ on at most a single user's input, i.e., $x_i = x'_i$ for all but one $i \in [n]$.

\begin{definition}[Differential Privacy (DP)~\citet{dwork2006calibrating,dwork2006our}]\label{def:dp_general}
Let $\epsilon, \delta \in \mathbb{R}_{\geq 0}$. A randomized algorithm $\mathcal{A}$ taking as input a dataset is said to be \emph{$(\epsilon, \delta)$-differentially private} ($(\epsilon, \delta)$-DP) if for any two neighboring datasets $X$ and $X'$, and for any subset $S$ of outputs of $\mathcal{A}$, it holds that $\Pr[\mathcal{A}(X) \in S] \le e^{\epsilon} \cdot \Pr[\mathcal{A}(X') \in S] + \delta$.
\end{definition}

\paragraph{Shuffle DP Model.} A protocol over $n$ inputs in the shuffle DP model~\cite{bittau17,erlingsson2019amplification,CheuSUZZ19} consists of three procedures. A \emph{local randomizer} takes an input $x_i$ and outputs a set of messages. The \emph{shuffler} takes the multisets output by the local randomizer applied to each of $x_1,\dots,x_n$, and produces a random permutation of the messages as output. Finally, the \emph{analyzer} takes the output of the shuffler and computes the output of the protocol. Privacy in the shuffle model is enforced on the output of the shuffler when a single input is changed.

\section{Generic Protocol Description}\label{sec:protocol}

Below we describe the protocol for $\Delta$-summation that is private in the shuffle DP model. In our protocol, the randomizer will send messages, each of which is an integer in $\{-\Delta, \dots, +\Delta\}$. The analyzer simply sums up all the incoming messages. The messages sent from the randomizer can be categorized into three classes:
\begin{itemize}[nosep]
\item \textbf{Input}: each user $i$ will send $x_i$ if it is non-zero. 
\item \textbf{Central Noise}: This is the noise whose sum is equal to the Discrete Laplace noise commonly used algorithms in the central DP model. 
This noise is sent in ``unary'' as $+1$ or $-1$ messages.
\item \textbf{Zero-Sum Noise}: Finally, we ``flood'' the messages with noise that cancels out. This noise comes from a carefully chosen sub-collection $\cS$ of the collection of all multisets of $\{-\Delta, \dots, +\Delta\} \setminus \{ 0 \}$ whose sum of elements is equal to zero (e.g., $\{-1, -1, +2\}$ may belong to $\cS$).\footnote{Note that while $\cS$ may be infinite, we will 
later set it to be finite, resulting in an efficient protocol.} For  more details, see Theorem~\ref{thm:right-inv} and the paragraph succeeding it.
We will refer to each $\bs \in \cS$ as a \emph{noise atom}. 
\end{itemize}

Algorithms~\ref{alg:delta_randomizer} and~\ref{alg:delta_analyzer} show the generic form of our protocol, which we refer to as the \emph{Correlated Noise} mechanism.
The protocol is specified by the following infinitely divisible distributions over $\Z_{\geq 0}$: the ``central'' noise distribution $\cDcentral$, and for every $\bs \in \cS$, the ``flooding'' noise distribution $\cD^\bs$.
%
\begin{algorithm}[t]
\caption{\small $\Delta$-Summation Randomizer} \label{alg:delta_randomizer}
\begin{algorithmic}[1]
\Procedure{\randomizer$_n(x_i)$}{}
\If{$x_i \ne 0$}
\SubState Send $x_i$
\EndIf
\State Sample $z^{+1}_i, z^{-1}_i \sim \cDcentral_{/n}$
\State Send $z^{+1}_i$ copies of $+1$, and $z^{-1}_i$ copies of $-1$ \label{line:plusoneminusone} 
\For{$\bs \in \cS$}
\SubState Sample $z^{\bs}_i \sim \cD^{\bs}_{/n}$ 
\SubFor{$m \in \bs$} \label{line:for-loop-multiset}
\SubSubState Send $z^{\bs}_i$ copies of $m$
\EndFor
\EndFor 
\EndProcedure
\end{algorithmic}
\end{algorithm}
\begin{algorithm}[t]
\caption{\small $\Delta$-Summation Analyzer} \label{alg:delta_analyzer}
\begin{algorithmic}[1]
\Procedure{\analyzer}{}
\State $R \leftarrow$ multiset of messages received
\State \Return $\sum_{y \in R} y$
\EndProcedure
\end{algorithmic}
\end{algorithm}

Note that since $\bs$ is a multiset, Line~\ref{line:for-loop-multiset} goes over each element the same number of times it appears in $\bs$; e.g., if $\bs = \{-1, -1, +2\}$, the iteration $m = -1$ is executed twice.

\subsection{Error and Communication Complexity}

We now state generic forms for the MSE and communication cost of the protocol: 


\begin{observation} \label{obs:error}
MSE is $2\Var(\cDcentral)$.
\end{observation}

We stress here that the distribution $\cDcentral$ itself is \emph{not} the Discrete Laplace distribution; we pick it so that $\cDcentral - \cDcentral$ is $\DLap$. As a result, $2\Var(\cDcentral)$ is indeed equal to the variance of the Discrete Laplace noise.

\begin{observation} \label{obs:comm}
Each user sends at most $1 + $ $\frac{1}{n}\left(2\E[\cDcentral] + \sum_{\bs \in \cS} |\bs| \cdot \E[\cD^{\bs}]\right)$ messages in expectation, each consisting of $\lceil \log \Delta \rceil + 1$ bits.
\end{observation}

\section{Parameter Selection and Privacy Proof}\label{sec:privacy_pf_and_params}

The focus of this section is on selecting concrete distributions to initiate the protocol and formalize its privacy guarantees, ultimately proving \Cref{thm:dsum-main}. First, in \Cref{sec:add-notations}, we introduce additional notation and reduce our task to proving a privacy guarantee for a protocol in the \emph{central} model. With these simplifications, we give a generic form of privacy guarantees in \Cref{sec:privacy-generic}.  \Cref{sec:nb} and \Cref{sec:right-inverse} are devoted to a more concrete selection of parameters. Finally, \Cref{thm:dsum-main} is proved in \Cref{sec:main-proof}.




\subsection{Additional Notation and Simplifications}
\label{sec:add-notations}


\paragraph{Matrix-Vector Notation.}
We use boldface letters to denote vectors and matrices, and standard letters to refer to their coordinates (e.g., if $\bu$ is a vector, then $u_i$ refers to its $i$th coordinate). For convenience, we allow general index sets for vectors and matrices; e.g., for an index set $\cI$, we write $\bu \in \R^{\cI}$ to denote the tuple $(u_i)_{i \in \cI}$. Operations such as addition, scalar-vector/matrix multiplication or matrix-vector multiplication are defined naturally.

For $i \in \cI$, we use $\bone_i$ to denote the $i$th vector in the standard basis; that is, its $i$-indexed coordinate is equal to  $1$ and each of the other coordinates is equal to $0$.  
Furthermore, we use~$\bzero$ to denote the all-zeros vector.

Let $[-\Delta, \Delta]$ denote $\{-\Delta, \dots, +\Delta\} \setminus \{ 0 \}$, 
and $\bv \in \Z^{[-\Delta, \Delta]}$ denote the vector $v_i := i$. Recall that a noise atom $\bs$ is a multiset of elements from $[-\Delta, \Delta]$. It is useful to also think of $\bs$ as a vector in  $\Z_{\geq 0}^{[-\Delta, \Delta]}$ where its $i$th entry denotes the number of times $i$ appears in $\bs$. We overload the notation and use $\bs$ to both represent the multiset and its corresponding vector. 

Let $\bA \in \Z^{[-\Delta, \Delta] \times \cS}$ denote the matrix whose rows are indexed by $[-\Delta, \Delta]$ and whose columns are indexed by $\cS$ where $A_{i, \bs} = s_i$. In other words, $\bA$ is a concatenation of column vectors $\bs$. 
Furthermore, let $[-\Delta, \Delta]_{-1}$ denote $[-\Delta, \Delta] \setminus \{1\}$, and $\tbA \in \Z^{[-\Delta, \Delta]_{-1} \times \cS}$ denote the matrix $\bA$ with row $1$ removed.

Next, we think of each input dataset $(x_1, \dots, x_n)$ as its histogram $\bh \in \Z_{\geq 0}^{[\Delta]}$ where $h_j$ denotes the number of $i \in [n]$ such that $x_i = j$. Under this notation, two input datasets $\bh, \bh'$ are neighbors iff $\|\bh - \bh'\|_{\infty} \leq 1$ and $\|\bh - \bh'\|_1 \leq 2$. For each histogram $\bh \in \Z_{\geq 0}^{[\Delta]}$, we write $\bh^{\pad} \in \Z_{\geq 0}^{[-\Delta, \Delta]}$ to denote the vector resulting from appending $\Delta$ zeros to the beginning of $\bh$; more formally, for every $i \in [-\Delta, \Delta]$, we let $h^{\pad}_i =  h_i$ if $i > 0$ and $h^{\pad}_i = 0$ if $i < 0$. 

\paragraph{An Equivalent Central DP Algorithm.}
A benefit of using infinitely divisible noise distributions is that they allow us to translate our protocols to equivalent ones in the \emph{central} model, where the total sum of the noise terms has a well-understood distribution. In particular, with the notation introduced above, \Cref{alg:delta_randomizer} corresponds to \Cref{alg:central-m} in the central model:

\begin{algorithm}[H]
\caption{\small Central Algorithm (Matrix-Vector Notation)} \label{alg:central-m}
\begin{algorithmic}[1]
\Procedure{\centralalgm($\bh$)}{}
\State Sample $z^{+1}, z^{-1} \sim \cDcentral$
\For{$\bs \in \cS$}
\SubState Sample $z^{\bs} \sim \cD^{\bs}$
\EndFor
\State $\bz \leftarrow (z^{\bs})_{\bs \in \cS}$
\State \Return $\bh^{\pad} + z^{+1} \cdot \bone_{1} + z^{-1} \cdot \bone_{-1} + \bA \bz$.
\EndProcedure
\end{algorithmic}
\end{algorithm}

\begin{observation} \label{obs:shuffle-to-central-matrix}
\randomizer\ is $(\eps, \delta)$-DP in the shuffle model if and only if \centralalgm\ is $(\eps, \delta)$-DP in the central model.
\end{observation}

Given \Cref{obs:shuffle-to-central-matrix}, we can focus on proving the privacy guarantee of \centralalgm\ in the central model, which will be the majority of this section.

\paragraph{Noise Addition Mechanisms for Matrix-Based Linear Queries.}
%
%
The \emph{$\cD$-noise addition mechanism} for $\Delta$-summation (as defined in Section~\ref{sec:intro}) works by first computing the summation and adding to it a noise random variable sampled from $\cD$, where $\cD$ is a distribution over integers.
Note that under our vector notation above, the $\cD$-noise addition mechanism simply outputs $\left<\bv, \bh\right> + z$ where $z \sim \cD$.


It will be helpful to consider a generalization of the $\Delta$-summation problem, which allows $\bv$ above to be changed to any matrix (where the noise is now also a vector).

To define such a problem formally, let $\cI$ be any index set. Given a matrix $\bQ \in \Z^{\cI \times [\Delta]}$, the \emph{$\bQ$-linear query} problem\footnote{Similar definitions are widely used in literature; see e.g.~\cite{NikolovTZ13}. The main distinction of our definition is that we only consider integer-valued $\bQ$ and $\bh$.} is to compute, given an input histogram $\bh \in \Z^{[\Delta]}_{\geq 0}$, an estimate of $\bQ\bh \in \Z^\cI$. (Equivalently, one can think of each user as holding a column vector of $\bQ$ or the all-zeros vector $\bzero$, and the goal is to compute the sum of these vectors.) 

The noise addition algorithms for $\Delta$-summation can be easily generalized to the $\bQ$-linear query case: for a collection $\cbD = (\cD^i)_{i \in \cI}$ of distributions,
the \emph{$\cbD$-noise addition mechanism} samples\footnote{Specifically, sample $z_i \sim \cD^i$ independently for each $i \in \cI$.} $\bz \sim \cbD$ and then outputs $\bQ\bh + \bz$.

\subsection{Generic Privacy Guarantee}
\label{sec:privacy-generic}

With all the necessary notation ready, we can now state our main technical  theorem, which gives a privacy guarantee in terms of a right inverse of the matrix $\tbA$:

\begin{theorem} \label{thm:privacy-main-generic}
Let $\bC \in \Z^{\cS \times [-\Delta, \Delta]_{-1}}$ denote any right inverse of $\tbA$ (i.e., $\tbA\bC = \bI$) whose entries are integers. Suppose that the following holds:
\begin{itemize}[nosep]
\item The $\hcD$-noise addition mechanism is $(\eps_1, \delta_1)$-DP for $\Delta$-summation.
\item The $(\tcD^{\bs})_{\bs \in \cS}$-noise addition mechanism is $(\eps_2, \delta_2)$-DP for the $\left[\bzero~\bc_2~\cdots~\bc_\Delta\right]$-linear query problem, where $\bc_i$ denotes the $i$th column of $\bC$. 
\end{itemize}
Then, \centralalgm\ with the following parameter selections is $(\eps^* + \eps_1 + \eps_2, \delta_1 + \delta_2)$-DP for $\Delta$-summation:
\begin{itemize}[nosep]
\item $\cDcentral = \NB(1, e^{-\eps^* / \Delta})$.
\item $\cD^{\{-1, +1\}} = \tcD^{\{-1, +1\}} + \hcD$.
\item $\cD^{\bs} = \tcD^{\bs}$ for all $\bs \in \cS \setminus \{\{-1, +1\}\}$.
\end{itemize}
\end{theorem}

The right inverse $\bC$ indeed represents the ``change of basis'' alluded to in the introduction. It will be specified in the next subsections along with the noise distributions $\hcD, (\tcD^{\bs})_{\bs \in \cS}$. 

As one might have noticed from \Cref{thm:privacy-main-generic}, $\{-1, +1\}$ is somewhat different that other noise atoms, as its noise distribution $\cD^{\{-1, +1\}}$ is the sum of $\hcD$ and $\tcD^{\{-1, +1\}}$. A high-level explanation for this is that our central noise is sent as $-1, +1$ messages and we would like to use the noise atom $\{-1, +1\}$ to ``flood out'' the correlations left in $-1, +1$ messages. A precise version of this statement is given below in \Cref{lem:correlated-coord}. We remark that the first output coordinate $z^{+1} - z^{-1} + \left<\bv, \bh\right>$ alone has exactly the same distribution as the $\DLap(\eps^*/\Delta)$-noise addition mechanism. The main challenge in this analysis is that the two coordinates are correlated through $z^{-1}$; indeed, this is where the $\tz^{-1, +1}$ random variable helps ``flood out'' the correlation.

\begin{lemma} \label{lem:correlated-coord}
Let $\cM_{\cor}$ be a mechanism that, on input histogram $\bh \in \Z_{\geq 0}^{[\Delta]}$, works as follows:
\begin{itemize}[nosep]
\item Sample $z^{+1}, z^{-1}$ independently from $\NB(1, e^{-\eps^* / \Delta})$.
\item Sample $\tz^{-1, +1}$ from $\hcD$.
\item Output $\left(z^{+1} - z^{-1} + \left<\bv, \bh\right>, z^{-1} + \tz^{-1, +1}\right)$.
\end{itemize}
If the $\hcD$-noise addition mechanism is $(\eps_1, \delta_1)$-DP for $\Delta$-summation, then $\cM_{\cor}$ is $(\eps^* + \eps_1, \delta_1)$-DP. 
\end{lemma}

\Cref{lem:correlated-coord} is a \emph{direct improvement} over the main analysis of~\citet{GKMP20-icml}, whose proof (which works only for $\Delta = 1$) requires the $\hcD$-noise addition mechanism to be $(\eps_1, \delta_1)$-DP for $O(\log(1/\delta)/\eps)$-summation; we remove this $\log(1/\delta)/\eps$ factor. Our novel insight is that 
when conditioned on $z^{+1} + z^{-1} = c$, rather than conditioned on $z^{+1} + z^{-1} = c - \kappa$ for some $c, \kappa \in \Z$, the distributions of $z^{-1}$ are quite similar, up to a ``shift'' of $\kappa$ and a multiplicative factor of $e^{-\kappa \eps^* / \Delta}$. This allows us to ``match'' the two probability masses and achieve the improvement. The full proof of \Cref{lem:correlated-coord} 
is deferred to SM.



Let us now show how \Cref{lem:correlated-coord} can be used to prove \Cref{thm:privacy-main-generic}. At a high-level, we run the mechanism $\cM_{\cor}$ from \Cref{lem:correlated-coord} and the $(\tcD^{\bs})_{\bs \in \cS}$-noise addition mechanism, and argue that we can use their output to construct the output for $\centralalgm$. The intuition behind this is that the first coordinate of the output of $\cM_{\cor}$ gives the weighted sum of the desired output, and the second coordinate gives the number of $-1$ messages used to flood the central noise. As for the $(\tcD^{\bs})_{\bs \in \cS}$-noise addition mechanism, since $\bC$ is a right inverse of $\tbA$, we can use them to reconstruct the number of messages $-\Delta, \dots, -1, 2, \dots, \Delta$. The number of 1 messages can then be reconstructed from the weighted sum and all the numbers of other messages. These ideas are encapsulated in the proof below.

\begin{proof}[Proof of \Cref{thm:privacy-main-generic}]
Consider $\cM_{\mysim}$ defined as follows:
\begin{enumerate}[nosep]
\item First, run $\cM_{\cor}$ from \Cref{lem:correlated-coord} on input histogram $\bh$ to arrive at an output $(b_{\mysum}, b_{\{-1, +1\}})$
\item Second, run $(\tcD^{\bs})_{\bs \in \cS}$-mechanism for $\left[\bzero~\bc_2~\cdots~\bc_\Delta\right]$-linear query on $\bh$ to get an output $\by = (y_{\bs})_{\bs \in \cS}$.
\item Output $\bu = (u_{-\Delta}, \dots, u_{-1}, u_1, \dots, u_{\Delta})$ computed by letting 
$\bw = \tbA (\by + b_{\{-1, 1\}} \cdot \bone_{\{-1, +1\}})$ and then
\begin{align*}
u_i =
\begin{cases}
w_i &\text{ if } i \ne 1, \\
b_{\mysum} - \sum_{i \in [-\Delta, \Delta]_{-1}} i \cdot w_i &\text{ if } i = 1.
\end{cases}
\end{align*}
\end{enumerate}

From \Cref{lem:correlated-coord} and our assumption on $\hcD$, $\cM_{\cor}$ is $(\eps^* + \eps_1, \delta_1)$-DP. By assumption that the $(\tcD^{\bs})_{\bs \in \cS}$-noise addition mechanism is $(\eps_2, \delta_2)$-DP for $\left[\bzero~\bc_2~\cdots~\bc_\Delta\right]$-linear query and by the basic composition theorem, the first two steps of $\cM_{\mysim}$ are $(\eps^* + \eps_1 + \eps_2, \delta_1 + \delta_2)$-DP. The last step of $\cM_{\mysim}$ only uses the output from the first two steps; hence, by the post-processing property of DP, we can conclude that $\cM_{\mysim}$ is indeed $(\eps^* + \eps_1 + \eps_2, \delta_1 + \delta_2)$. 

Next, we claim that $\cM_{\mysim}(\bh)$ has the same distribution as $\centralalgm(\bh)$ with the specified parameters. To see this, recall from $\cM_{\cor}$ that we have $b_{\mysum} = \left<\bv, \bh\right> + z^{+1} - z^{-1}$, and $b_{\{-1, +1\}} = z^{-1} + \tz^{\{-1, +1\}}$, where $z^{+1}, z^{-1} \sim \NB(1, e^{-\eps^*/\Delta})$ and $\tz^{\{-1, +1\}} \sim \hcD$ are independent. Furthermore, from the definition of the $(\tcD^{\bs})_{\bs \in \cS}$-noise addition mechanism, we have
\begin{align*}
Y = \left[\bzero~\bc_2~\cdots~\bc_\Delta\right] \bh + \bbf
= \bC \tbh^{\pad} + \bbf,
\end{align*}
where $f_{\bs} \sim \tcD^{\bs}$ are independent, and $\tbh^{\pad}$ denotes $\bh^{\pad}$ after replacing its first coordinate with zero.

Notice that $\tbA \bone_{\{-1, +1\}} = \bone_{-1}$. Using this and our assumption that $\tbA\bC = \bI$, we get
\begin{align*}
\bw &= \tbA\left(\bC \tbh^{\pad} + \bbf + b_{\{-1, 1\}} \cdot \bone_{\{-1, +1\}}\right) \\
&= \tbh^{\pad} + z^{-1} \cdot \bone_{-1} + \tbA(\bbf + \tz^{\{-1, +1\}} \cdot \bone_{\{-1, +1\}}).
\end{align*}
Let $\bz = \bbf + \tZ^{\{-1, +1\}} \cdot \bone_{\{-1, +1\}}$; we can see that each entry $z_{\bs}$ is independently distributed as $\cD^{\bs}$. Finally, we have
\begin{align*}
u_1 
&= \left<\bv, \bh\right> + z^{+1} - z^{-1} - \sum_{i \in [-\Delta, \Delta]_{-1}} i \cdot w_i \\
&= \left(\sum_{i \in [\Delta]} i \cdot h_i\right) + z^{+1} - z^{-1} 
- \sum_{i \in [-\Delta, \Delta]_{-1}} i \cdot \left(\tbh^{\pad} + z^{-1} \cdot \bone_{-1} + \tbA\bz\right)_i \\
&= \left(\sum_{i \in [\Delta]} i \cdot h_i\right) + z^{+1} - z^{-1} - \left(\sum_{i \in \{2, \dots, \Delta\}} i \cdot h_i\right) 
+ z^{-1} - \left(\sum_{i \in [-\Delta, \Delta]_{-1}} i \cdot \left(\tbA\bz\right)_i\right) \\
&= h_1 + z^{+1} - \left(\sum_{i \in [-\Delta, \Delta]_{-1}} i \cdot \left(\tbA\bz\right)_i\right) \\
&= h_1 + z^{+1} - \left(\sum_{i \in [-\Delta, \Delta]_{-1}} i \cdot \left(\sum_{\bs \in \cS} s_i \cdot z_{\bs}\right)\right) \\
&= h_1 + z^{+1} - \left(\sum_{\bs \in \cS} z_{\bs} \cdot \left(\sum_{i \in [-\Delta, \Delta]_{-1}} i \cdot s_i \right)\right).
\end{align*}
Recall that $\sum_{i \in [-\Delta, \Delta]} i \cdot s_i = 0$ for all $\bs \in \cS$; equivalently, $\sum_{i \in [-\Delta, \Delta]_{-1}} i \cdot s_i = -s_1$. Thus, we have
\begin{align*}
u_1 = h_1 + z^{+1} + \left(\sum_{\bs \in \cS} z_{\bs} \cdot s_1\right) = h_1 + z^{+1} + (\bA\bz)_1.
\end{align*}
Hence, we can conclude that $\bu = \bh^{\pad} + z^{+1} \cdot \bone_{1} + z^{-1} \cdot \bone_{-1} + \bA \bz$; this implies that $\cM_{\mysim}(\bh)$ has the same distribution as the mechanism $\centralalgm(\bh)$.
\end{proof}

\subsection{Negative Binomial Mechanism}
\label{sec:nb}

Having established a generic privacy guarantee of our algorithm, we now have to specify the distributions $\hcD, \tcD^{\bs}$ that satisfy the conditions in \Cref{thm:privacy-main-generic} while keeping the number of messages sent for the noise small. Similar to~\citet{GKMP20-icml}, we use the negative binomial distribution. Its privacy guarantee is summarized below.\footnote{For the exact statement we use here, please refer to Theorem 13 of \citet{GKMP20-icml-arxiv}, which contains a correction of calculation errors in Theorem 13 of~\citet{GKMP20-icml}.}
\begin{theorem}[\citet{GKMP20-icml}] \label{thm:nb-scalar}
For any $\eps, \delta \in (0, 1)$ and $\Delta \in \N$, let $p = e^{-0.2\eps/\Delta}$ and $r = 3(1  + \log(1/\delta))$. The $\NB(r, p)$-additive noise mechanism is $(\eps, \delta)$-DP for $\Delta$-summation.
\end{theorem}


We next extend \Cref{thm:nb-scalar} to the $\bQ$-linear query problem. To state the formal guarantees, 
we say that a vector $\bx \in \Z^\cI$ is \emph{dominated} by vector $\by \in \N^\cI$ iff $\sum_{i \in \cI} |x_i| / y_i \leq 1$. (We use the convention $0/0 = 0$ and $a/0 = \infty$ for all $a > 0$.)

\begin{corollary} \label{cor:nb}
Let $\eps, \delta \in (0, 1)$.
Suppose that every column of $\bQ \in \Z^{\cI \times [\Delta]}$ is dominated by $\bt \in \N^\cI$. For each $i \in \cI$, let $p_i = e^{-0.2\eps/(2\bt_i)}$, $r_i = 3(1 + \log(|\cI|/\delta))$ and $\cD^i = \NB(r_i, p_i)$. Then, $(\cD^i)_{i \in \cI}$-noise addition mechanism is $(\eps, \delta)$-DP for $\bQ$-linear query.
\end{corollary}

\subsection{Finding a Right Inverse}
\label{sec:right-inverse}

A final step before we can apply \Cref{thm:privacy-main-generic} is to specify the noise atom collection $\cS$ and the right inverse $\bC$ of $\tbA$. Below we give such a right inverse where every column is dominated by a vector $\bt$ that is ``small''. This allows us to then use the negative binomial mechanism in the previous section with a ``small'' amount of noise. How ``small'' $\bt$ is depends on the expected number of messages sent; this is governed by $\|\bt\|_{\cS} := \sum_{\bs \in \cS} \|\bs\|_1 \cdot |\bt_{\bs}|$. With this notation, the guarantee of our right inverse can be stated as follows. 
(We note that in our noise selection below, every $\bs \in \cS$ has at most three elements. In other words, $\|\bt\|_{\cS}$ and $\|\bt\|_1$ will be within a factor of three of each other.)

\begin{theorem} \label{thm:right-inv}
There exist $\cS$ of size $O(\Delta)$, $\bC \in \Z^{\cS \times [-\Delta, \Delta]_{-1}}$ with $\tbA\bC = \bI$ and $\bt \in \N^\cS$ such that $\|\bt\|_{\cS} \leq O(\Delta \log^2 \Delta)$ and every column of $\bC$ is dominated by $\bt$.
\end{theorem}

The full proof of \Cref{thm:right-inv} is deferred to SM. The main idea is to essentially proceed via Gaussian elimination on $\tbA$. However, we have to be careful about our choice of orders of rows/columns to run the elimination on, as otherwise it might produce a non-integer matrix $\bC$ or one whose columns are not ``small''. In our proof, we order the rows based on their absolute values, and we set $\cS$ to be the collection of $\{-1, +1\}$ and $\{i, -\lceil i/2 \rceil, -\lfloor i/2 \rfloor\}$ for all $i \in \{-\Delta, \dots, -2, 2, \dots, \Delta\}$. In other words, these are the noise atoms we send in our protocol.

\begin{figure*}[h]
\centering
\begin{subfigure}[b]{0.32\textwidth}
\includegraphics[width=\textwidth]{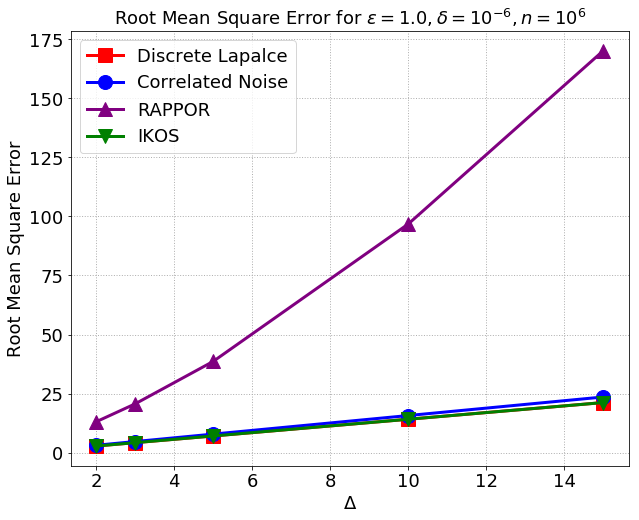}
\caption{RMSEs for varying $\Delta$}
\label{subfig:error}
\end{subfigure}
\begin{subfigure}[b]{0.32\textwidth}
\includegraphics[width=\textwidth]{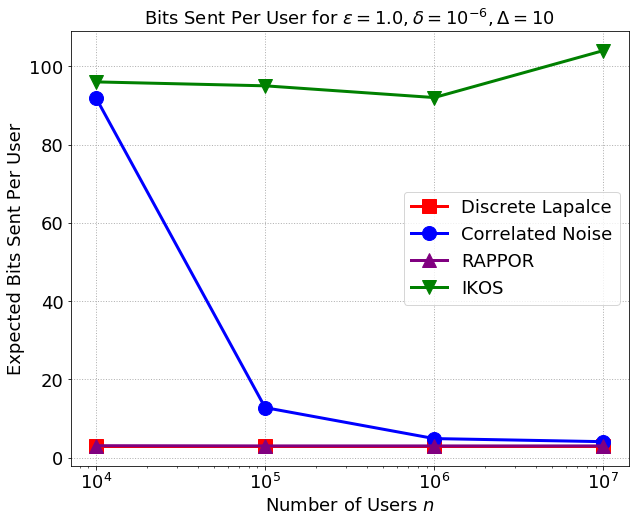}
\caption{Expected bits sent for varying $n$}
\label{subfig:comm-varying-n}
\end{subfigure}
\begin{subfigure}[b]{0.32\textwidth}
\includegraphics[width=\textwidth]{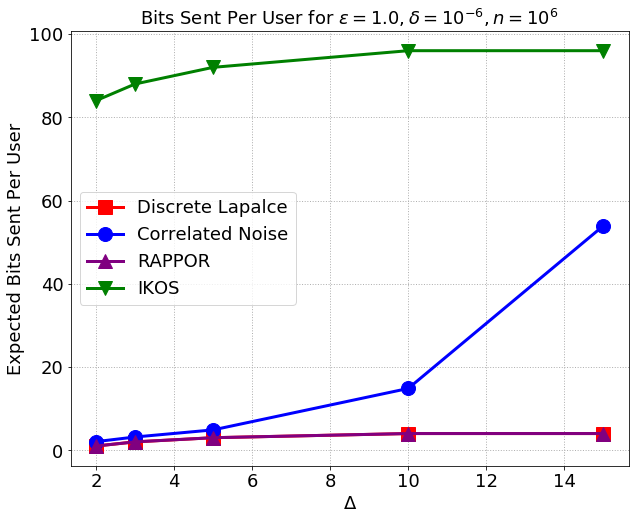}
\caption{Expected bits sent for varying $\Delta$}
\label{subfig:comm-varying-Delta}
\end{subfigure}
\caption{Error and communication complexity of our ``correlated noise'' $\Delta$-summation protocol compared to other protocols. 
}
\end{figure*}

\subsection{Specific Parameter Selection: Proof of \Cref{thm:dsum-main}}
\label{sec:main-proof}

\begin{proof}[Proof of \Cref{thm:dsum-main}]
Let $\cS, \bC, \bt$ be as in Theorem~\ref{thm:right-inv}, $\eps^* = (1 - \gamma)\eps, \eps_1 = \eps_2 = \min\{1, \gamma\eps\}/2, \delta_1 = \delta_2 = \delta/2$, and,
\begin{itemize}[nosep]
\item $\cDcentral = \NB(1, e^{-\eps^*/\Delta})$,
\item $\hcD = \NB(\hat{r}, \hat{p})$ where $\hat{p} = e^{-0.2\eps_1/\Delta}$ and $\hat{r} = 3(1  + \log(1/\delta_1))$,
\item $\tcD^{\bs} = \NB(r_{\bs}, p_{\bs})$ where $p_{\bs} = e^{-0.2\eps_2/(2 t_{\bs})}$ and $r^{\bs} = 3(1 + \log(|\cS|/\delta_2))$ for all $\bs \in \cS$.
\end{itemize}
From \Cref{thm:nb-scalar}, the $\hcD$-noise addition mechanism is $(\eps_1, \delta_1)$-DP for $\Delta$-summation. \Cref{cor:nb} implies that the $(\tcD^{\bs})_{\bs \in \cS}$-noise addition  mechanism is $(\eps_2, \delta_2)$-DP for $[\bzero~\bc_2~\cdots~\bc_\Delta]$-linear query. As a result, since $\eps^* + \eps_1 + \eps_2 \leq \eps$ and $\delta_1 + \delta_2 = \delta$, \Cref{thm:privacy-main-generic} ensures that \centralalgm\ with parameters as specified in the theorem is $(\eps, \delta)$-DP.

The error claim follows immediately from \Cref{obs:error} together with the fact that $\NB(1, e^{-\eps^*/\Delta}) - \NB(1, e^{-\eps^*}/\Delta) = \DLap(\eps^*/\Delta)$. Using \Cref{obs:comm}, we can also bound the expected number of messages as
\begin{align*}
&1 + \frac{1}{n}\left(2\E[\cDcentral] + \sum_{\bs \in \cS} |\bs| \cdot \E[\cD^{\bs}]\right) \\
&\leq 1 + O\left(\frac{\Delta}{\eps^* n}\right) + \frac{1}{n} \left(\frac{2\hat{r}}{1 - \hat{p}} + \sum_{\bs \in \cS} \frac{|\bs| r_{\bs}}{1 - p_{\bs}}\right) \\
&= 1 + O\left(\frac{\Delta \log(1/\delta)}{\gamma \eps n}\right) + O\left(\frac{\|\bt\|_{\cS} \cdot \log(\Delta/\delta)}{\eps_2 n}\right) \\
&\leq 1 + O\left(\frac{\Delta \log(1/\delta)}{\gamma \eps n}\right) + O\left(\frac{\Delta \log^2 \Delta \cdot \log(\Delta/\delta)}{\gamma \eps n}\right),
\end{align*}
where each message consists of $\lceil \log \Delta \rceil + 1$ bits.
\end{proof}

\section{Experimental Evaluation}\label{sec:experiments}

We compare our ``correlated noise'' $\Delta$-summation protocol (Algorithms~\ref{alg:delta_randomizer},~\ref{alg:delta_analyzer}) against known algorithms in the literature, namely the IKOS ``split-and-mix'' protocol~\cite{ishai2006cryptography,balle_merged,ghazi2019private} and the fragmented\footnote{The RAPPOR randomizer starts with a one-hot encoding of the input (which is a $\Delta$-bit string) and flips each bit with a certain probability. \emph{Fragmentation} means that, instead of sending the entire $\Delta$-bit string, the randomizer only sends the coordinates that are set to 1, each as a separate $\lceil \log \Delta \rceil$-bit message. This is known to reduce both the communication and the error in the shuffle model~\cite{CheuSUZZ19,esa-revisited}.} version of RAPPOR~\cite{erlingsson2014rappor,esa-revisited}. We also include the Discrete Laplace mechanism~\cite{GhoshRS12} in our plots for comparison, although it is not directly implementable in the shuffle model. We do not include the generalized (i.e., $\Delta$-ary) Randomized Response algorithm~\cite{warner1965randomized} as it always incurs at least as large error as RAPPOR.

For our protocol, we set $\eps^*$ in \Cref{thm:privacy-main-generic} to $0.9 \eps$, meaning that its MSE is that of $\DLap(0.9 \eps/\Delta)$. While the parameters set in our proofs give a theoretically vanishing overhead guarantee, they turn out to be rather impractical; instead, we resort to a tighter numerical approach to find the parameters. We discuss this, together with the setting of parameters for the other alternatives, in the SM.

For all mechanisms the root mean square error (RMSE) and the (expected) communication per user only depends on $n, \eps, \delta, \Delta$, and is independent of the input data. We next summarize our findings.

\paragraph{Error.}
The IKOS algorithm has the same error as the Discrete Laplace mechanism (in the central model), whereas our algorithm's error is slightly larger due to $\eps^*$ being slightly smaller than $\eps$.
On the other hand, the RMSE for RAPPOR grows as $1/\delta$ increases, but it seems to converge as $n$ becomes larger. (For $\delta = 10^{-6}$ and $n = 10^4, \dots, 10^7$, we found that RMSEs differ by less than 1\%.)

However, the key takeaway here is that the RMSEs of IKOS, Discrete Laplace, and our algorithm grow only linearly in $\Delta$, but the RMSE of RAPPOR is proportional to $\Theta(\Delta^{3/2})$. This is illustrated in \Cref{subfig:error}.

\paragraph{Communication.}
While the IKOS protocol achieves the same accuracy as the Discrete Laplace mechanism, it incurs a large communication overhead, as each message sent consists of $\lceil \log (n\Delta) \rceil$ bits and each user needs to send multiple messages. By contrast, when fixing $\Delta$ and taking $n \to \infty$, both RAPPOR and our algorithm only send $1 + o(1)$ messages, each of length $\lceil \log \Delta \rceil$ and $\lceil \log \Delta \rceil +  1$ respectively. This is illustrated in \Cref{subfig:comm-varying-n}. Note that the number of bits sent for IKOS is indeed not a monotone function since, as $n$ increases, the number of messages required decreases but the length of each message increases.

Finally, we demonstrate the effect of varying $\Delta$ in \Cref{subfig:comm-varying-Delta} for a fixed value of $n$. Although \Cref{thm:dsum-main} suggests that the communication overhead should grow roughly linearly with $\Delta$, we have observed larger gaps in the experiments. This seems to stem from the fact that, while the $\tilde{O}(\Delta)$ growth would have been observed if we were using our analytic formula, our tighter parameter computation (detailed in Appendix~\ref{subsec:noise-parameter-search} of SM) finds a protocol with even \emph{smaller} communication, suggesting that the actual growth might be less than $\tilde{O}(\Delta)$ though we do not know of a formal proof of this. Unfortunately, for large $\Delta$, the optimization problem for the parameter search gets too demanding and our program does not find a good solution, leading to the ``larger gap'' observed in the experiments.

\section{Conclusions}\label{sec:conclusion}
In this work, we presented a DP protocol for real and $1$-sparse vector aggregation in the shuffle model with accuracy arbitrarily close to the best possible central accuracy, 
and with relative communication overhead tending to~$0$ with increasing number of users. It would be very interesting to generalize our protocol and obtain qualitatively similar guarantees for \emph{dense} vector summation.

We also point out that in the low privacy regime ($\eps \gg 1$), the \emph{staircase} mechanism is known to significantly improve upon the Laplace mechanism~\cite{GengV16,GengKOV15}; the former achieves MSE 
that is exponentially small in $\Theta(\eps)$
while the latter has MSE $O(1/\eps^2)$. An interesting open question is to achieve such a gain in the shuffle model. 

\balance
\bibliographystyle{icml2021}
\bibliography{refs}

\newpage
\onecolumn

\appendix


\section*{\centerline{\huge Supplementary Material}}

\section{Additional Preliminaries}

We start by introducing a few additional notation and lemmas that will be used in our proofs.

\begin{definition}
The \emph{$\eps$-hockey stick divergence} of two (discrete) distributions $\cD, \cD'$ is defined as
\begin{align*}
d_{\eps}(\cD \| \cD') := \sum_{v \in \supp(\cD)} \left[\Pr_{x \sim \cD}[x = v] - e^{\eps} \cdot \Pr_{x \sim \cD'}[x = v]\right]_+ 
\end{align*}
where $[y]_+ := \max\{y, 0\}$.
\end{definition}

The following is a (well-known) simple restatement of DP in terms of hockey stick divergence:

\begin{lemma} \label{lem:dp-hockeystick}
An algorithm $\mathcal{A}$ is $(\epsilon, \delta)$-DP iff for any neighboring datasets $X$ and $X'$, it holds that $d_{\eps}(\cA(X), \cA(X')) \leq \delta$.
\end{lemma}

We will also use the following result in~\citep[Lemma 21]{GKMP20-icml}, which allows us to easily compute differential privacy parameters noise addition algorithms for $\Delta$-summation.

\begin{lemma}[\citet{GKMP20-icml}] \label{lem:noiseaddition-hockey-stick}
An $\cD$-noise addition mechanism for $\Delta$-summation is $(\eps, \delta)$-DP iff $d_{\eps}(\cD \| k + \cD) \leq \delta$ for all $k \in \{-\Delta, \dots, +\Delta\}$.
\end{lemma}

Finally, we will also use the following lemma, which can be viewed as an alternative formulation of the basic composition theorem.

\begin{lemma} \label{lem:basic-composition-hockeystick}
Let $\cD_1, \dots, \cD_k, \cD'_1, \dots, \cD'_k$ be any distributions and $\eps_1, \dots, \eps_k$ any non-negative real numbers. Then, we have $d_{\eps_1 + \cdots + \eps_k}(\cD_1 \times \cdots \times \cD_k \| \cD'_1 \times \cdots \times \cD'_k) \leq \sum_{i \in [k]} d_{\eps_i}(\cD_i \| \cD'_i)$.
\end{lemma}

\section{On the Near Optimality of (Discrete) Laplace Mechanism}

\begin{figure}[h]
\centering
\includegraphics[width=0.5\textwidth]{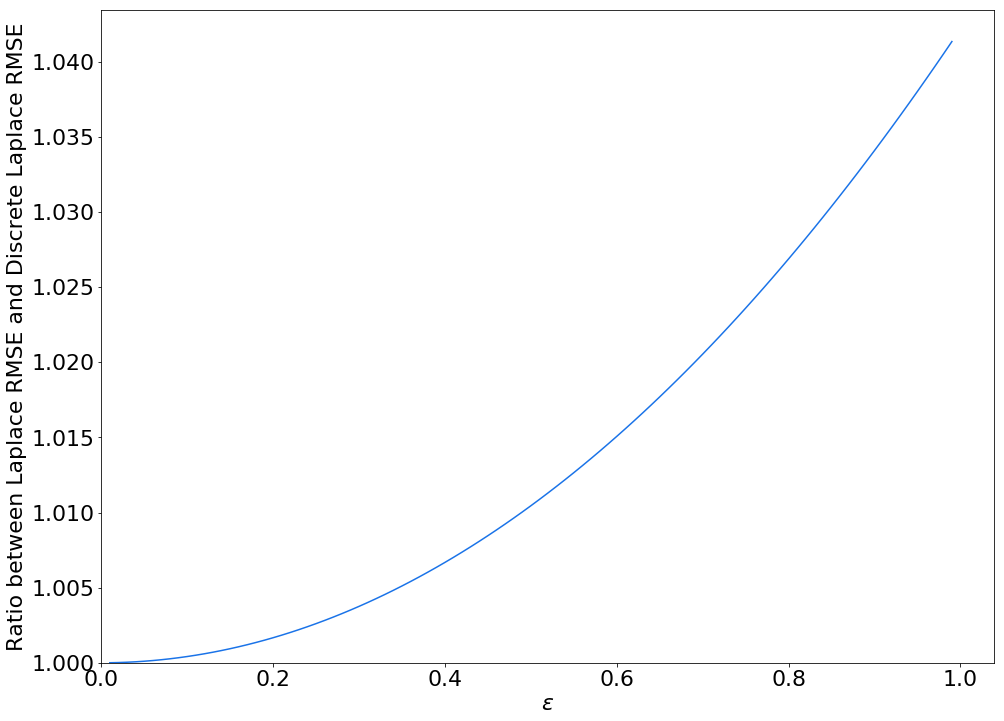}
\caption{Comparison between RMSE of the Laplace Mechanism and RMSE of the Discrete Laplace Mechanism when varing $\eps$.}
\label{fig:laplace-v-dlaplace}
\end{figure}

In this section, we briefly discuss the near-optimality of Laplace mechanism in the central model. First, we recall that it follows from~\citet{GhoshRS12} that the smallest MSE of any $\eps$-DP algorithm (in the central model) for binary summation (i.e., 1-summation) is at least that of the Discrete Laplace mechanism minus a term that converges to zero as $n \to \infty$. In other words, RMSE of the Discrete Laplace mechanism, up to an $o(1)$ additive factor, forms a lower bound for binary summation. Notice that a lower bound on the RMSE for binary summation also yields a lower bound on the RMSE for real summation, since the latter is a more general problem. We plot the comparison between RMSE of the Discrete Laplace mechanism and RMSE of the Laplace mechanism in \Cref{fig:laplace-v-dlaplace} for $\eps \in [0, 1]$. Since the former is a lower bound on RMSE of any $\eps$-DP algorithm for real summation, this plot shows how close the Laplace mechanism is to the lower bound. Specifically, we have found that, in this range of $\eps < 1$, the Laplace mechanism is within $5\%$ of the optimal error (assuming $n$ is sufficiently large). Finally, note that our real summation algorithm (\Cref{thm:real-main}) can achieve accuracy arbitrarily close to the Laplace mechanism, and hence our algorithm can also get to within $5\%$ of the optimal error for any $\eps \in [0, 1]$.

\section{Missing Proofs from \Cref{sec:protocol}}

Below we provide (simple) proofs of the generic form of the error (\Cref{obs:error}) and the communication cost (\Cref{obs:comm}) of our protocol.

\begin{proof}[Proof of \Cref{obs:error}]
Consider each user $i$. The sum of its output is not affected by the noise from $\cS$ since the sum of their messages are zero; in other words, the sum is simply $x_i + Z^{+1}_i - Z^{-1}_i$. As a result, the total error is exactly equal to $\left(\sum_{i \in [n]} Z^{+1}_i\right) - \left(\sum_{i \in [n]} Z^{-1}_i\right)$. Since each of $Z^{+1}_i, Z^{-1}_i$ is an i.i.d. random variable distributed as $\cDcentral_{/n}$, the error is distributed as $\cDcentral - \cDcentral$.
\end{proof}

\begin{proof}[Proof of \Cref{obs:comm}]
Each user sends at most one message corresponding to its input; in addition to this, the expected number of messages sent in the subsequent steps is equal to
\begin{align*}
\E\left[Z^{+1}_i + Z^{-1}_i + \sum_{\bs \in \cS} \sum_{m \in \bs} Z^{\bs}_i\right]
&= \E[\cDcentral_{/n}] + \E[\cDcentral_{/n}] + \sum_{\bs \in \cS} \sum_{m \in \bs} \E[\cD^{\bs}_{/n}] \\
&= 2\E[\cDcentral_{/n}] + \sum_{\bs \in \cS} |\bs| \cdot \E[\cD^{\bs}_{/n}] \\
&= \frac{1}{n}\left(2\E[\cDcentral] + \sum_{\bs \in \cS} |\bs| \cdot \E[\cD^{\bs}]\right)
\end{align*}
Finally, since each message is a number from $[-\Delta, \Delta]$, it can be represented using $\lceil \log_2 \Delta \rceil + 1$ bits as desired.
\end{proof}

\section{Missing Proofs from \Cref{sec:privacy_pf_and_params}}

\subsection{From Shuffle Model to Central Model}
\label{sec:shuffle-v-central}

Let us restate \Cref{alg:central-m} in a slightly different notation where each user's input $x_1, \dots, x_n$ is represented explicitly in the algorithm (instead of the histogram as in \Cref{alg:central-m}):
\begin{algorithm}[H]
\caption{\small Central Algorithm} \label{alg:central}
\begin{algorithmic}[1]
\Procedure{\centralalg($x_1, \dots, x_n$)}{}
\SubState $u_{-\Delta}, \dots, u_{-1} \leftarrow 0$
\For{$j \in [\Delta]$}
\SubState $u_j \leftarrow |\{i \in [n] \mid x_i = j\}|$
\EndFor
\State Sample $z^{+1} \sim \cDcentral$
\State $u_{+1} \leftarrow u_{+1} + Z^1$
\State Sample $z^{-1} \sim \cDcentral$
\State $u_{-1} \leftarrow u_{-1} + Z^{-1}$
\For{$\bs \in \cS$}
\SubState Sample $z^{\bs} \sim \cD^{\bs}$
\SubFor{$m \in \bs$}
\SubSubState $u_m \leftarrow u_m + z^{\bs}$
\EndFor
\EndFor 
\State \Return $(u_{-\Delta}, \dots, u_{-1}, u_{+1}, \dots,  u_{+\Delta})$.
\EndProcedure
\end{algorithmic}
\end{algorithm}

\begin{observation} \label{obs:shuffle-to-central}
\randomizer\ is $(\eps, \delta)$-DP in the shuffle model iff \centralalg\ is $(\eps, \delta)$-DP in the central model.
\end{observation}

\begin{proof}
Consider the analyzer's view for \randomizer; all the analyzer sees is a number of messages $-\Delta, \dots, -1, +1, \dots, +\Delta$. Let $u_i$ denote the number of messages $i$, for a non-zero integer $i \in [-\Delta, \Delta]$. We will argue that $(u_{-\Delta}, \dots, u_{-1}, u_{+1}, \dots,  u_{+\Delta})$  has the same distribution as the output of \centralalg, from which the observation follows. To see this, notice that the total number of noise atoms $\bs$ in \randomizer\ across all users is $z^{\bs}_1 + \cdots + z^{\bs}_n$ which is distributed as $\cD^{\bs}$. Similarly, the number of additional $-1$ messages and that of $+1$ messages (from Line~\ref{line:plusoneminusone} of \randomizer) are both distributed as (independent) $\cDcentral$ random variables. As a result, $(u_{-\Delta}, \dots, u_{-1}, u_{+1}, \dots,  u_{+\Delta})$  has the same distribution as the output of \centralalg.
\end{proof}


\Cref{obs:shuffle-to-central-matrix} then immediately follows from \Cref{obs:shuffle-to-central} and the fact that the two algorithms are equivalent.

\subsection{Flooding Central Noise: Proof of \Cref{lem:correlated-coord}}

\begin{proof}[Proof of \Cref{lem:correlated-coord}]
Consider any two input neighboring datasets $\bh, \bh'$. Let $\eps := \eps^* + \eps_1$. From \Cref{lem:dp-hockeystick}, it suffices to show that $d_{\eps}(\cM_{\cor}(\bh), \cM_{\cor}(\bh')) \leq \delta_1$. Let $\kappa = \left<\bv, \bh\right> - \left<\bv, \bh'\right>$; note that by definition of neighboring datasets, we have $-\Delta \leq \kappa \leq \Delta$. We may simplify $d_{\eps}(\cM_{\cor}(\bh), \cM_{\cor}(\bh'))$ as
\begin{align}
&d_{\eps}(\cM_{\cor}(\bh), \cM_{\cor}(\bh')) \nonumber \\
&= \sum_{a, b \in \Z} \Big[\Pr[z^{+1} - z^{-1} + \left<\bv, \bh\right> = a, z^{-1} + \tz^{-1, +1} = b] - e^\eps \Pr[z^1 - z^{-1} + \left<\bv, \bh'\right> = a, z^{-1} + \tz^{-1, +1} = b]\Big]_+ \nonumber \\
&= \sum_{c, b \in \Z} \Big[\Pr[z^{+1} - z^{-1} = c - \kappa, z^{-1} + \tz^{-1, +1} = b] - e^\eps \Pr[z^{+1} - z^{-1} = c, z^{-1} + \tz^{-1, +1} = b]\Big]_+, \label{eq:correlated-hs-div}
\end{align}
where the second equality comes from simply substituting $c = a - \left<\bv, \bh'\right>$.

For convenience, let 
$\lambda_c := \max\{0, \kappa - c\}$. We may expand the first probability above as
\begin{align}
&\Pr[z^{+1} - z^{-1} = c - \kappa, z^{-1} + \tz^{-1, +1} = b] \nonumber \\
&= \sum_{d = \lambda_c}^{\infty} \Pr[z^{+1} = c - \kappa + d]\Pr[z^{-1} = d]\Pr[\tz^{\{-1, +1\}} = b - d]. \label{eq:correlated-first-prob-expanded}
\end{align}
Similarly, let $\gamma_c := \max\{0, c\}$. We may expand the second probability above as
\begin{align}
&\Pr[z^{+1} - z^{-1} = c, z^{-1} + \tz^{-1, +1} = b] \nonumber \\
&= \sum_{d' = \gamma_c}^{\infty} \Pr[z^{+1} = c + d']\Pr[z^{-1} = d']\Pr[\tz^{\{-1, +1\}} = b - d'] \nonumber \\
&= \sum_{d = \lambda_c}^{\infty} \Pr[z^{+1} = c + d - (\lambda_c - \gamma_c)]\Pr[z^{-1} = d - (\lambda_c - \gamma_c)]\Pr[\tz^{\{-1, +1\}} = b - d + (\lambda_c - \gamma_c)]. \label{eq:correlated-second-prob-expanded}
\end{align} 
Next, notice that for $d \geq \lambda_c$, we have
\begin{align*}
&\frac{\Pr[z^{+1} = c + d - (\lambda_c - \gamma_c)]\Pr[z^{-1} = d - (\lambda_c - \gamma_c)]}{\Pr[z^{+1} = c - \kappa + d]\Pr[z^{-1} = d]} \\
&= \frac{\exp\left(-\frac{\eps^*}{\Delta}(c + d - (\lambda_c - \gamma_c))\right)\exp\left(-\frac{\eps^*}{\Delta}(d - (\lambda_c - \gamma_c))\right)}{\exp\left(-\frac{\eps^*}{\Delta}(c - \kappa + d)\right)\exp\left(-\frac{\eps^*}{\Delta}(d)\right)} \\
&= \exp\left(-\frac{\eps^*}{\Delta}(\kappa - 2(\lambda_c - \gamma_c))\right) \\
&\geq \exp(-\eps^*),
\end{align*}
where the last inequality employs a simple observation that $(\lambda_c - \gamma_c) \leq \kappa \leq \Delta$.

Plugging the above back into \eqref{eq:correlated-second-prob-expanded}, we get
\begin{align}
&\Pr[z^{+1} - z^{-1} = c, z^{-1} + \tz^{-1, +1} = b] \\
&\geq e^{-\eps^*} \cdot \sum_{d = \lambda_c}^{\infty} \Pr[z^{+1} = c - \kappa + d]\Pr[z^{-1} = d]\Pr[\tz^{\{-1, +1\}} = b - d + (\lambda_c - \gamma_c)]. \label{eq:correlated-second-prob-expanded2}
\end{align}
Then, plugging both \eqref{eq:correlated-first-prob-expanded} and \eqref{eq:correlated-second-prob-expanded2} into \eqref{eq:correlated-hs-div}, we arrive at
\begin{align*}
&d_{\eps}(\cM_{\cor}(\bh), \cM_{\cor}(\bh')) \\
&\leq \sum_{c, b \in \Z} \left[ \sum_{d = \lambda_c}^{\infty} \Pr[z^{+1} = c - \kappa + d]\Pr[z^{-1} = d]  \right. \\
& \qquad \qquad \left(\Pr[\tz^{\{-1, +1\}} = b - d] - e^{\eps_1} \Pr[\tz^{\{-1, +1\}} = b - d + (\lambda_c - \gamma_c)]\right) \Bigg]_+ \\
&\leq \sum_{c, b \in \Z}\sum_{d = \lambda_c}^{\infty} \Pr[z^{+1} = c - \kappa + d]\Pr[z^{-1} = d] \\
& \qquad\qquad \cdot \left[\Pr[\tz^{\{-1, +1\}} = b - d] - e^{\eps_1} \Pr[\tz^{\{-1, +1\}} = b - d + (\lambda_c - \gamma_c)]\right]_+ \\
&= \sum_{c \in \Z} \sum_{d = \lambda_c}^{\infty} \Pr[z^{+1} = c - \kappa + d]\Pr[z^{-1} = d] \\
& \qquad\qquad \left(\sum_{b \in \Z} \left[\Pr[\tz^{\{-1, +1\}} = b - d] - e^{\eps_1} \Pr[\tz^{\{-1, +1\}} = b - d + (\lambda_c - \gamma_c)]\right]_+ \right) \\
&= \sum_{c \in \Z} \sum_{d = \lambda_c}^{\infty} \Pr[z^{+1} = c - \kappa + d]\Pr[z^{-1} = d] \cdot d_{\eps_1}(\hcD \| \hcD + (\gamma_c - \lambda_c)).
\end{align*}
Now, recall that $|\gamma_c - \lambda_c| \leq |\kappa| \leq \Delta$. Thus, from our assumption that the $\hcD$-mechanism is $(\eps_1, \delta_1)$-DP for $\Delta$-summation and \Cref{lem:noiseaddition-hockey-stick}, we can conclude that $d_{\eps_1}(\hcD \| \hcD + (\gamma_c - \lambda_c)) \leq \delta_1$. Hence, we can further bound $d_{\eps}(\cM_{\cor}(\bh), \cM_{\cor}(\bh'))$ as
\begin{align*}
d_{\eps}(\cM_{\cor}(\bh), \cM_{\cor}(\bh')) \leq \delta_1 \cdot \left(\sum_{c \in \Z} \sum_{d = \lambda_c}^{\infty} \Pr[z^{+1} = c - \kappa + d]\Pr[z^{-1} = d]\right) = \delta_1.
\end{align*}
This implies that $\cM_{\cor}$ is $(\eps^* + \eps_1, \delta_1)$-DP as desired.
\end{proof}

\subsection{Finding a Right Inverse: Proof of \Cref{thm:right-inv}}

\begin{proof}[Proof of \Cref{thm:right-inv}]
We first select $\cS$ to be the collection $\{-1, +1\}$ and all multisets of the form $\{i, -\lfloor i/2 \rfloor, -\lceil i/2 \rceil\}$ for all $i \in \{-\Delta, \dots, -2, +2, \dots, +\Delta\}$.

Throughout this proof, we write $\bc_i$ to denote the $i$th column of $\bC$. We construct a right inverse $\bC$ by constructing its columns iteratively as follows:
\begin{itemize}
\item First, we let $\bc_{-1} = \bone_{\{-1, +1\}}$ denote the indicator vector of coordinate $\bs = \{-1, +1\}$.
\item For convenience, we let $\bc_1 = \bzero$ although this column is not present in the final matrix $\bC$.
\item For each $i \in \{2, \dots, \Delta\}$:
\begin{itemize}
\item Let $i_1 = \lfloor i / 2 \rfloor$ and $i_2 = \lceil i / 2 \rceil$.
\item Let $\bc_i$ be $\bone_{\{i, -i_1, -i_2\}} - \bc_{-i_1} - \bc_{-i_2}$.
\item Let $\bc_{-i}$ be $\bone_{\{-i, i_1, i_2\}} - \bc_{i_1} - \bc_{i_2}$.
\end{itemize}
\end{itemize}

We will now show by induction that $\bC$ as constructed above indeed satisfies $\tbA\bC = \bI$. Let $\tbA_{i, *}$ denote the $i$th row of $\tbA$. Equivalently, we would like to show that\footnote{We use $\bone[E]$ to denote the indicator of condition $E$, i.e., $\bone[E] = 1$ if $E$ holds and $\bone[E] = 0$ otherwise.}
\begin{align} \label{eq:inverse-dot-product-condition}
\tbA_{i', *}\bc_{i} = \bone[i' = i],
\end{align}
for all $[-\Delta, \Delta]_{-1}$.
We will prove this by induction on $|i|$.

\paragraph{(Base Case)} We will show that this holds for $i = -1$.  Since $\bc_1 = \bone_{\{-1, 1\}}$, we have
\begin{align*}
\tbA_{i', *}\bc_i = \tbA_{i, \{-1, 1\}} = \bone[i \in \{-1, 1\}] = \bone[i = -1].
\end{align*}

\paragraph{(Inductive Step)} Now, suppose that~\eqref{eq:inverse-dot-product-condition} holds for all $i$ such that $|i| < m$ for some $m \in \N \setminus \{1\}$. We will now show that it holds for $i = m, -m$ as well. Let $m_1 = \lfloor m/2 \rfloor$ and $m_2 = \lceil m/2 \rceil$. For $i = m$, let $\bs = \{m, -m_1, -m_2\}$; we have
\begin{align*}
\bA_{i', *} \bc_i 
&= \bA_{i', *} \left(\bone_{\bs} - \bc_{-m_1} - \bc_{-m_2}\right) \\
\text{(From inductive hypothesis)} &= \bA_{i', *} \bone_{\bs} - \bone[i' = -m_1] - \bone[i' = -m_2] \\
&= \bs_{i'} - \bone[i' = -m_1] - \bone[i' = -m_2] \\
&= \bone[i = i'],
\end{align*}
where the last inequality follows from the definition of $\bs$.

The case where $i = -m$ is similar. Specifically, letting $\bs = \{-m, m_1, m_2\}$, we have
\begin{align*}
\bA_{i', *} \bc_i 
&= \bA_{i', *} \left(\bone_{\bs} - \bc_{m_1} - \bc_{m_2}\right) \\
\text{(From inductive hypothesis)} &= \bA_{i', *} \bone_{\bs} - \bone[i' = m_1] - \bone[i' = m_2] \\
&= \bs_{i'} - \bone[i' = m_1] - \bone[i' = m_2] \\
&= \bone[i = i'].
\end{align*}

Let $\Gamma = \Delta \cdot \lceil 1 + \log \Delta \rceil$.
Let $\bt \in \Z^\cS$ be defined by
\begin{align*}
\bt_{\bs} =
\begin{cases}
\Gamma &\text{ if } \bs = \{-1, 1\} \\
\lceil \Gamma / m \rceil &\text{ if } \bs = \{m, -\lceil m/2 \rceil, -\lfloor m / 2 \rfloor\} \text{ or } \bs = \{-m, \lceil m/2 \rceil, \lfloor m / 2 \rfloor\} \text{ for some } m \in \{2, \dots, \Delta\}.
\end{cases}
\end{align*}
Since $\|\bs\|_1 \leq 3$ for all $\bs \in \cS$, we have
\begin{align*}
\|\bt\|_{\cS} \leq 3 \cdot \|\bt\|_1 = \Gamma + \sum_{m=2,\dots,\Delta} \lceil \Gamma / m \rceil = \Gamma \cdot O(\log \Delta) = O(\Delta \log^2 \Delta),
\end{align*}
as desired. 

Finally, we will show via induction on $|i|$ that
\begin{align} \label{eq:bounded-weighted-sensitivity}
\sum_{\bs \in \cS} \frac{|\bC_{\bs, i}|}{\bt_\bs} \leq \frac{|i| \cdot \lceil 1 + \log |i| \rceil}{\Gamma},
\end{align}
which, from our choice of $\Gamma$, implies the last property.

\paragraph{(Base Case)}
We have
\begin{align*}
\sum_{\bs \in \cS} \frac{|\bC_{\bs, -1}|}{\bt_\bs} = \frac{1}{\bt_{\{-1, +1\}}} = \frac{1}{\Gamma}.
\end{align*}

\paragraph{(Inductive Step)}
For convenience, in the derivation below, we think of $\bC_{\bs, 1}$ as 0 for all $\bs \in \cS$; note that in fact, row $1$ does not exist for $\bC$.

Suppose that~\eqref{eq:bounded-weighted-sensitivity} holds for all $i \in [-\Delta, \Delta]_{-1}$ such that $|i| < m$ for some $m \in \N \setminus \{1\}$. Now, consider $i = m$; letting $m_1 = \lceil m/2 \rceil, m_2 = \lfloor m/2 \rfloor$from the definition of $\bc_i$, we have
\begin{align*}
\sum_{\bs \in \cS} \frac{|\bC_{\bs, i}|}{\bt_\bs}
&\leq \frac{1}{\bt_{\{m, -m_1, -m_2\}}} + \sum_{\bs \in \cS} \frac{|\bC_{\bs, m_1}|}{\bt_\bs} + \sum_{\bs \in \cS} \frac{|\bC_{\bs, m_2}|}{\bt_\bs} \\
(\text{From inductive hypothesis}) &\leq \frac{1}{\bt_{\{m, -m_1, -m_2\}}} + \frac{m_1 \lceil 1 + \log m_1 \rceil}{\Gamma} + \frac{m_2 \lceil 1 + \log m_2 \rceil}{\Gamma} \\
&\leq \frac{1}{\bt_{\{m, -m_1, -m_2\}}} + \frac{m_1 \lceil \log m \rceil}{\Gamma} + \frac{m_2 \lceil \log m \rceil}{\Gamma} \\
&= \frac{1}{\bt_{\{m, -m_1, -m_2\}}} + \frac{m \lceil \log m \rceil}{\Gamma} \\
(\text{From our choice of } \bt_{\{m, -m_1, -m_2\}}) &\leq \frac{m}{\Gamma} + \frac{m \lceil \log m \rceil}{\Gamma} \\
&= \frac{m \lceil 1 + \log m \rceil}{\Gamma}.
\end{align*}
Similarly, for $i = -m$, we have
\begin{align*}
\sum_{\bs \in \cS} \frac{|\bC_{\bs, i}|}{\bt_\bs}
&\leq \frac{1}{\bt_{\{-m, m_1, m_2\}}} + \sum_{\bs \in \cS} \frac{|\bC_{\bs, -m_1}|}{\bt_\bs} + \sum_{\bs \in \cS} \frac{|\bC_{\bs, -m_2}|}{\bt_\bs} \\
(\text{From inductive hypothesis}) &\leq \frac{1}{\bt_{\{-m, m_1, m_2\}}} + \frac{m_1 \lceil 1 + \log m_1 \rceil}{\Gamma} + \frac{m_2 \lceil 1 + \log m_2 \rceil}{\Gamma} \\
&\leq \frac{1}{\bt_{\{-m, m_1, m_2\}}} + \frac{m_1 \lceil \log m \rceil}{\Gamma} + \frac{m_2 \lceil \log m \rceil}{\Gamma} \\
&= \frac{1}{\bt_{\{-m, m_1, m_2\}}} + \frac{m \lceil \log m \rceil}{\Gamma} \\
(\text{From our choice of } \bt_{\{-m, m_1, m_2\}}) &\leq \frac{m}{\Gamma} + \frac{m \lceil \log m \rceil}{\Gamma} \\
&= \frac{m \lceil 1 + \log m \rceil}{\Gamma}.
\qedhere
\end{align*}

\end{proof}

\subsection{Negative Binomial Mechanism for Linear Queries: Proof of \Cref{cor:nb}}

\begin{proof}[Proof of \Cref{cor:nb}]
We denote the $(\cD^i)_{i \in \cI}$-noise addition mechanism by $\cM$. Furthermore, we write $\cbD$ to denote the noise distribution (i.e., distribution of $\bz$ where $z_i \sim \cD^i$).

Consider two neighboring datasets $\bh, \bh'$. We have
\begin{align*}
d_{\eps}(\cM(\bh) \| \cM(\bh'))
&= d_{\eps}\left(\bQ\bh + \cbD \middle\| \bQ\bh' + \cbD\right) \\
&= d_{\eps}\left(\cbD \| \bQ(\bh' - \bh) + \cbD\right) \\
&= d_{\eps}\left(\prod_{i \in \cI} \cD^i \middle\| \prod_{i \in \cI} ((\bQ(\bh' - \bh))_i + \cD^i)\right).
\end{align*}
Now, notice that $\bQ(\bh' - \bh)$ is dominated by $2\bt$. Let $\eps_i := \eps \cdot \frac{|(\bQ(\bh' - \bh))_i|}{2\bt_i}$ for all $i \in \cI$; we have $\sum_{i \in \cI} \eps_i \leq \eps$. As a result, applying \Cref{lem:basic-composition-hockeystick}, we get
\begin{align*}
d_{\eps}(\cM(\bh) \| \cM(\bh')) 
&\leq \sum_{i \in \cI} d_{\eps_i}(\cD^i \| (\bQ(\bh' - \bh))_i + \cD^i) \\
(\text{From \Cref{thm:nb-scalar} and our choice of } \eps_i, p_i, r_i) &\leq \sum_{i \in \cI} \delta / |\cI|  \\
&= \delta,
\end{align*}
which concludes our proof.
\end{proof}

\subsection{From Integers to Real Numbers: Proof of \Cref{thm:real-main}}

\label{sec:real-sum}

We now prove \Cref{thm:real-main}. The proof follows the discretization via randomized rounding of~\cite{balle_merged}.

\begin{proof}[Proof of \Cref{thm:real-main}]
The algorithm for real summation works as follows:
\begin{itemize}
\item Let $\Delta = \lceil \frac{\eps}{2} \cdot \sqrt{n/\zeta} \rceil$.
\item Each user randomly round the input $x_i \in [0, 1]$ to $y_i \in \{0, \dots, \Delta\}$ such that
\begin{align*}
y_i =
\begin{cases}
\lfloor x_i \Delta \rfloor &\text{ with probability } 1 - (x_i - \lfloor x_i \Delta \rfloor), \\
\lfloor x_i \Delta \rfloor + 1 &\text{ with probability }  x_i - \lfloor x_i \Delta \rfloor. \\
\end{cases}
\end{align*}
\item Run $\Delta$-summation algorithm from \Cref{thm:dsum-main} on the inputs $y_1, \dots, y_n$ with $\gamma = \zeta / 2$ to get an estimated sum $s$.
\item Output $s/\Delta$.
\end{itemize}

Using the communication guarantee of \Cref{thm:dsum-main}, we can conclude that, in the above algorithm, each user sends 
\begin{align*}
1 + \left(\frac{\Delta}{\gamma \eps n}\right) \cdot O\left(\log(1/\delta) + \log^2 \Delta \cdot \log(\Delta/\delta)\right) = 1 + \left(\frac{1}{\zeta \sqrt{\zeta n}}\right) O\left(\log(1/\delta) + \log^2(n/\zeta) \cdot \log(n/(\zeta\delta))\right)
\end{align*}
messages in expectation and that each message contains
\begin{align*}
\lceil \log \Delta \rceil + 1 = \left\lceil \log \left(\frac{\eps}{2} \cdot \sqrt{n/\zeta}\right) \right\rceil + 1 \leq 0.5(\log n + \log(1/\zeta)) + O(1)
\end{align*}
bits.

We will next analyze the MSE of the protocol. As guaranteed by \Cref{thm:dsum-main}, $s$ can be written as $y_1 + \cdots + y_n + e$ where $e$ is distributed as $\DLap((1 - \gamma)\eps/\Delta)$. Since $e, \Delta x_1 - y_1, \dots, \Delta x_n - y_n$ are all independent and have zero mean, the MSE can be rearranged as
\begin{align*}
\E[(x_1 + \cdots + x_n - s/\Delta)^2]
&= \frac{1}{\Delta^2} \cdot \E[((\Delta x_1 - y_1) + \cdots + (\Delta x_n - y_n) - e)^2] \\
&= \frac{1}{\Delta^2} \left(\E[e^2] + \sum_{i=1}^{n} \E[(\Delta x_i - y_i)^2] \right).
\end{align*}
Recall that $\E[e^2]$ is simply the variance of the discrete Laplace distribution with parameter $(1 - \gamma)\eps/\Delta$, which is at most $\frac{2\Delta^2}{(1 - \gamma)^2\eps^2}$. Moreover, for each $i$, $\E[(\Delta x_i - y_i)^2] = (x_i - \lfloor x_i \Delta \rfloor)(1 - (x_i - \lfloor x_i \Delta \rfloor)) \leq 1/4$. Plugging this to the above equality, we have
\begin{align*}
\E[(x_1 + \cdots + x_n - s/\Delta)^2] = \frac{1}{\Delta^2} \left(\frac{2\Delta^2}{(1 - \gamma)^2\eps^2} + \frac{n}{4}\right) \leq \frac{2}{(1 - \zeta)^2\eps^2},
\end{align*}
where the inequality follows from our parameter selection.
The right hand side of the above inequality is indeed the MSE of Laplace mechanism with parameter $(1 - \zeta)\eps$. This concludes our proof.
\end{proof}

\section{From Real Summation to 1-Sparse Vector Summation: Proof of \Cref{cor:1-sparse}}

In this section, we prove our result for 1-Sparse Vector Summation (\Cref{cor:1-sparse}). The idea, formalized below, is to run our real-summation randomizer in each coordinate independently.

\begin{proof}[Proof of \Cref{cor:1-sparse}]
Let $v^i \in \R^d$ denote the (1-sparse) vector input to the $i$th user. The randomizer simply runs the real summation randomizer on each coordinate of $v^i$, which gives the messages to send. It then attaches to each message the coordinate, as to distinguish the different coordinates. We stress that the randomizer for each coordinate is run with privacy parameters $(\eps/2, \delta/2)$ instead of $(\eps, \delta)$.

\begin{algorithm}[h]
\caption{\small 1-Sparse Vector Randomizer}
\begin{algorithmic}[1]
\Procedure{\sparsevectorrandomizer$_n(v^i)$}{}
\For{$j \in [d]$}
\SubState $S_j \leftarrow \randomizer^{\eps/2, \delta/2}(v^i_j)$
\SubFor{$m \in S_j$}
\SubSubState Send $(i, m)$
\EndFor
\EndFor
\EndProcedure
\end{algorithmic}
\end{algorithm}

The analyzer simply separates the messages based on the coordinates attached to them. Once this is done, messages corresponding to each coordinate are then plugged into the real summation randomizer (that just sums them up), which gives us the estimate of that coordinate. (Note that, similar to $R$, each $R_j$ is a multiset.)

\begin{algorithm}[h]
\caption{\small 1-Sparse Vector Randomizer} \label{alg:sparse_randomizer}
\begin{algorithmic}[1]
\Procedure{\sparsevectoranalyzer}{}
\State $R \leftarrow$ multiset of messages received
\For{$j \in [d]$}
\SubState $R_j \leftarrow \{m \mid (j, m) \in R_j\}$ 
\EndFor
\State \Return $(\analyzer(R_1), \dots, \analyzer(R_d))$
\EndProcedure
\end{algorithmic}
\end{algorithm}

The accuracy claim follows trivially from that of the real summation accuracy in \Cref{thm:real-main} with $(\eps/2, \delta/2)$-DP. In terms of the communication complexity, since $v^i$ is non-zero in only one coordinate and that the expected number of messages sent for \randomizer\ with zero input is only $\tilde{O}_{\zeta, \epsilon}\left(\frac{\log(1/\delta)}{\sqrt{n}}\right)$, the total expected number of messages sent by \sparsevectorrandomizer\ is
\begin{align*}
1 + d \cdot \tilde{O}_{\zeta, \epsilon}\left(\frac{\log(1/\delta)}{\sqrt{n}}\right)
\end{align*}
as desired. Furthermore, each message is the real summation message, which consists of only $\tfrac{1}{2} \log n + O(\log \frac{1}{\zeta})$ bits, appended with a coordinate, which can be represented in $\lceil \log d\rceil$ bits. Hence, the total number of bits required to represent each message of \sparsevectorrandomizer\ is $\log{d} + \tfrac{1}{2} \log n + O(\log \frac{1}{\zeta})$.

We will finally prove that the algorithm is $(\eps/2, \delta/2)$-DP. Consider any two neighboring input datasets $X = (v_1, \dots, v_n)$ and $X' = (v_1, \dots, v_{n - 1}, v'_n)$ that differs only on the last user's input. Let $J \subseteq [d]$ denote the set of coordinates that $v_n, v'_n$ differ on. Since $v_n, v'_n$ are both 1-sparse, we can conclude that $|J| \leq 2$. Notice that the view of the analyzer is $R_1, \dots, R_d$; let $\cD_1, \dots, \cD_d$ be the distributions of $R_1, \dots, R_d$ respectively for the input dataset $X$, and $\cD'_1, \dots, \cD'_d$ be the respective distributions for the input dataset $X'$. We have
\begin{align*}
d_{\eps}(\cD_1 \times \cdots \times \cD_d \| \cD'_1 \times \cdots \times \cD'_d) &= d_{\eps}\left(\prod_{j \in J} \cD_j \middle\| \prod_{j \in J} \cD'_j\right) \\
(\text{From \Cref{lem:basic-composition-hockeystick} and } |J| \leq 2) &\leq \sum_{j \in J} d_{\eps/2}\left(\cD_j \| \cD'_j\right) \\
(\text{Since \randomizer\ is } (\eps/2,\delta/2)\text{-DP}) &\leq \sum_{i \in J} \delta/2 \\
&\leq \delta.
\end{align*}
From this and since $R_1, \dots, R_d$ are independent, we can conclude that the algorithm is $(\eps, \delta)$-DP as claimed.
\end{proof}

\section{Concrete Parameter Computations}
\label{sec:concrete-param}

In this section, we describe modifications to the above proofs that result in a more practical set of parameters. As explained in the next section, these are used in our experiments to get a smaller amount of noise than the analytic bounds.

\subsection{Tight Bound for Matrix-Specified Linear Queries}

By following the definition of differential privacy, one can arrive at a compact description of when the mechanism is differentially private:
\begin{lemma}
The $\cbD$-noise addition mechanism is $(\eps, \delta)$-DP for $\bQ$-linear query problem if the following holds for every pair of columns $j, j' \in [\Delta]$ of $\bQ$:
\begin{align*}
d_{\eps}\left(\prod_{i \in \cI} (\cD^i + Q_{i, j} - Q_{i, j'}) \middle\| \prod_{i \in \cI} \cD^i\right) \leq \delta.
\end{align*}
\end{lemma}

\begin{proof}
This follows from the fact that the output distribution is $\prod_{i \in \cI} (\cD^i + \bQ_{i, *} \bh)$, $d_{\eps}$ is shift-invariant, and \Cref{lem:dp-hockeystick}.
\end{proof}

Plugging the above into the negative binomial mechanism, we get:
\begin{lemma} \label{lem:nb-refined}
Suppose that $\cD^i = \NB(r_i, p_i)$ for each $i \in \cI$.
The $(\cD^i)_{\cI}$-mechanism is $(\eps, \delta)$-DP for $\bQ$-linear query problem if the following holds for every pair of columns $i, i' \in [q]$ of $\bQ$:
\begin{align*}
d_{\eps}\left(\prod_{i \in \cI} (\NB(r_i, p_i) + \bQ_{i, j} - \bQ_{i, j'}) \middle\| \prod_{i \in \cI} \NB(r_i, p_i)\right) \leq \delta.
\end{align*}
\end{lemma}

\section{Refined Analytic Bounds}

We use the following analytic bound, which is a more refined version of \Cref{cor:nb}. Notice below that we may select $\bt' = 2\bt$ where $\bt$ is as in \Cref{cor:nb} and satisfies the requirement. However, the following bound allows us greater flexibility in choosing $\bt'$. By optimizing for $\bt'$, we get 30-50\% reduction in the noise magnitude compared to using \Cref{cor:nb} together with the explicitly constructed $\bt$ from the proof of \Cref{thm:right-inv}.

\begin{lemma} \label{lem:analytic-refined}
Suppose that, for every pair of columns $\bq_j, \bq_{j'}$ of $\bQ \in \Z^{\cI \times [\Delta]}$, $\bq_j - \bq_{j'}$ is dominated by $\bt' \in \R^\cI_{+}$. For each $i \in \cI$, let $p_i = e^{-0.2\eps/ t'_i)}$, $r_i = 3(1 + \log(|\cI|/\delta))$ and $\cD^i = \NB(r_i, p_i)$. Then, $(\cD^i)_{i \in \cI}$-noise addition mechanism is $(\eps, \delta)$-DP for $\bQ$-linear query.
\end{lemma}

\begin{proof}
We denote the $(\cD^i)_{i \in \cI}$-noise addition mechanism by $\cM$. Furthermore, we write $\cbD$ to denote the noise distribution (i.e., distribution of $\bz$ where $z_i \sim \cD^i$).

Consider two neighboring datasets $\bh, \bh'$. We have
\begin{align*}
d_{\eps}(\cM(\bh) \| \cM(\bh'))
&= d_{\eps}\left(\bQ\bh + \cbD \middle\| \bQ\bh' + \cbD\right) \\
&= d_{\eps}\left(\cbD \| \bQ(\bh' - \bh) + \cbD\right) \\
&= d_{\eps}\left(\prod_{i \in \cI} \cD^i \middle\| \prod_{i \in \cI} ((\bQ(\bh' - \bh))_i + \cD^i)\right).
\end{align*}
Now, notice that our assumption implies that $\bQ(\bh' - \bh)$ is dominated by $\bt'$. Let $\eps_i := \eps \cdot \frac{|(\bQ(\bh' - \bh))_i|}{t'_i}$ for all $i \in \cI$; we have $\sum_{i \in \cI} \eps_i \leq \eps$. As a result, applying \Cref{lem:basic-composition-hockeystick}, we get
\begin{align*}
d_{\eps}(\cM(\bh) \| \cM(\bh')) 
&\leq \sum_{i \in \cI} d_{\eps_i}(\cD^i \| (\bQ(\bh' - \bh))_i + \cD^i) \\
(\text{From \Cref{thm:nb-scalar} and our choice of } \eps_i, p_i, r_i) &\leq \sum_{i \in \cI} \delta / |\cI|  \\
&= \delta,
\end{align*}
which concludes our proof.
\end{proof}

\section{Experimental Evaluation: Additional Details and Results}

In this section, we give additional details on how to set the parameters and provide additional experimental results, including the effects of $\eps, \delta$ on the $\Delta$-summation protocols and an experiment on the real summation problem on a census dataset.

\subsection{Noise parameter settings.}
\label{subsec:noise-parameter-search}

\paragraph{Our Correlated Noise Protocol.}
Recall that the notions of $\cDcentral, \hcD, \tcD^{\bs}, \eps^*, \eps_1, \eps_2, \delta_1, \delta_2$ from \Cref{thm:dsum-main}. We always set $\eps^*$ beforehand: in the $\Delta$-summation experiments (where the entire errors come from the Discrete Laplace noises), $\eps^*$ is set to be $0.9\eps$ as to minimize the noise. On the other hand, for real summation experiments where most of the errors are from discretization, we set $\eps^*$ to be $0.1\eps$; this is also to help reduce the number of messages as there is now more privacy budget allocated for the $\hcD$ and $\tcD^{\bs}$ noise. Once $\eps^*$ is set, we split the remaining $\eps - \eps^*$ among $\eps_1$ and $\eps_2$; we give between 50\% and 90\% to $\eps_1$ depending on the other parameters, so as to minimize the number of messages. 

Once the privacy budget is split, we first use Lemma~\ref{lem:analytic-refined} to set $\hcD = \NB(\hr, \hp), \tcD^{\bs} = \NB(r^{\bs}, p^{\bs})$; the latter also involves an optimization problem over $\bt'$, which turns out to be solvable in a relatively short amount of time for moderate values of $\Delta$. This gives us the ``initial'' parameters for $\br, \hp, r^{\bs}, p^{\bs}$. We then formulate a generic optimization problem, together with the condition in \Cref{lem:nb-refined}, and use it to further search for improved parameters. This last search for parameters works well for small values of $\Delta$ and results in as much as $40-80\%$ reduction in the expected number of messages. However, as $\Delta$ grows, this gain becomes smaller since the number of parameters grows linearly with $\Delta$ and the condition in \Cref{lem:nb-refined} also takes more resources to compute; thus, the search algorithm fails to yield parameters significantly better than the initial parameters.

\paragraph{Fragmented RAPPOR.}
We use the same approach as \citet{GKMP20-icml} to set the parameter (i.e., flip probability) of Fragmented RAPPOR. At a high level, we set a parameter in an ``optimistic'' manner, meaning that the error and expected number of messages reported might be even lower than the true flip probability that is $(\eps, \delta)$-DP. Roughly speaking, when we would like to check whether a flip probability is feasible, we take two explicit input datasets $X, X'$ and check whether $d_{\eps}(\cM(X) \| \cM(X')) \leq \delta$. This is ``optimistic'' because, even if the parameter passes this test, it might still not be $(\eps, \delta)$-DP due to some other pair of datasets. For more information, please refer to Appendix G.3 of \citet{GKMP20-icml}.

\subsection{Effects of $\eps$ and $\delta$ on $\Delta$-Summation Results}

Recall that in \Cref{sec:experiments}, we consider the RMSEs and the communication complexity as $\Delta$ changes, and for the latter also as $n$ changes. In this section, we extend the study to the effect of $\eps, \delta$.  Here, we fix $n = 10^6$ and $\Delta = 5$ throughout. When $\delta$ varies, we fix $\eps = 1$; when $\eps$ varies, we fix $\delta = 10^{-6}$.

\paragraph{Error.} In terms of the error, RAPPOR's error slowly increases as $\delta$ decreases, while other algorithms' errors remain constant. As $\eps$ decreases, all algorithms' errors also increase. These are shown in \Cref{fig:err-varying-eps-delta}.

\begin{figure*}
\centering
\begin{subfigure}[b]{0.4\textwidth}
\includegraphics[width=\textwidth]{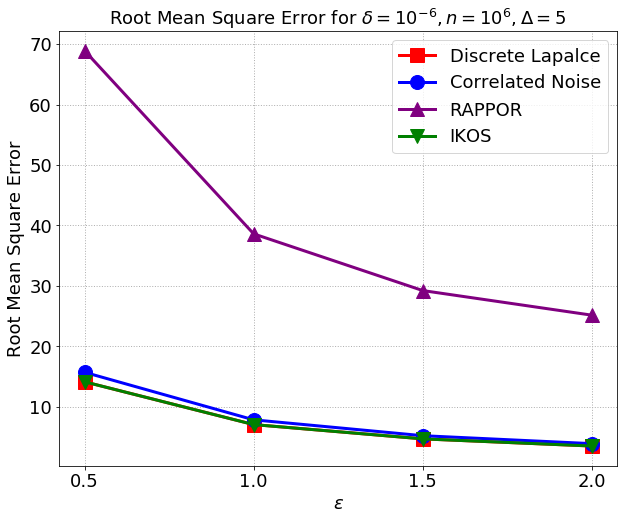}
\caption{RMSEs for varying $\eps$}
\end{subfigure}
\begin{subfigure}[b]{0.4\textwidth}
\includegraphics[width=\textwidth]{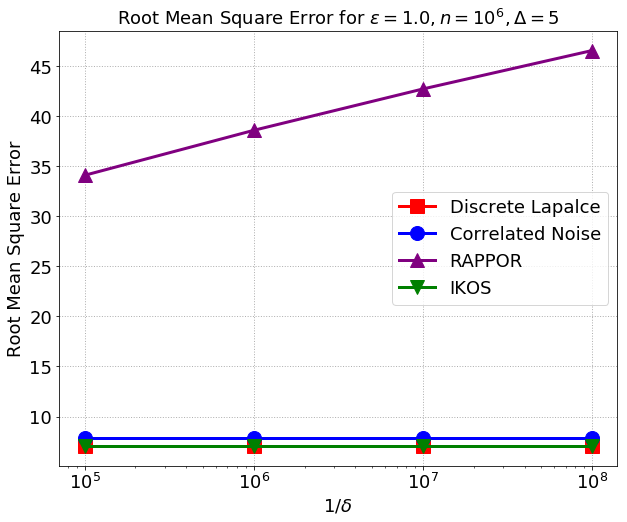}
\caption{RMSEs for varying $\delta$}
\end{subfigure}
\caption{Error of our ``correlated noise'' real-summation protocol compared to other protocols on the rent dataset when varying $\eps$ or $\delta$.}
\label{fig:err-varying-eps-delta}
\end{figure*}

\paragraph{Communication.} The number of bits required for the central DP algorithm remains constant (i.e., $\lceil \log \Delta \rceil$) regardless of the values of $\delta$ or $\eps$. In all other algorithms, the expected number of bits sent per user decreases when $\delta$ or $\eps$ increases. These are demonstrated in \Cref{fig:comm-varying-eps-delta}. We remark that this also holds for RAPPOR, even thought it might be hard to see from the plots because, in our regime of parameters, RAPPOR is within 0.1\% of the central DP algorithm.

\begin{figure*}
\centering
\begin{subfigure}[b]{0.4\textwidth}
\includegraphics[width=\textwidth]{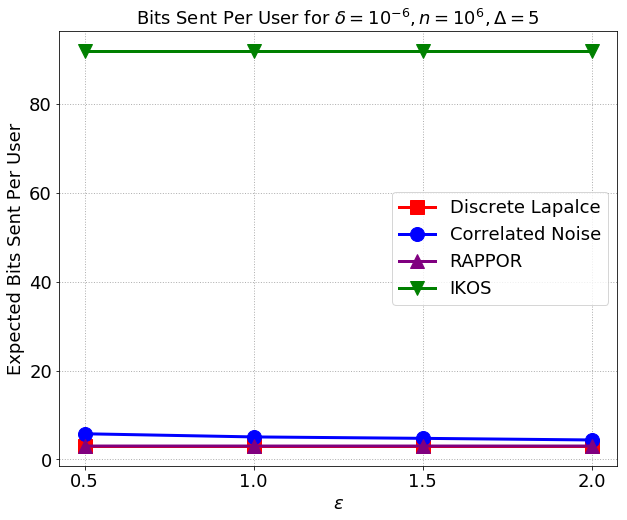}
\caption{Expected bits sent for varying $\eps$}
\end{subfigure}
\begin{subfigure}[b]{0.4\textwidth}
\includegraphics[width=\textwidth]{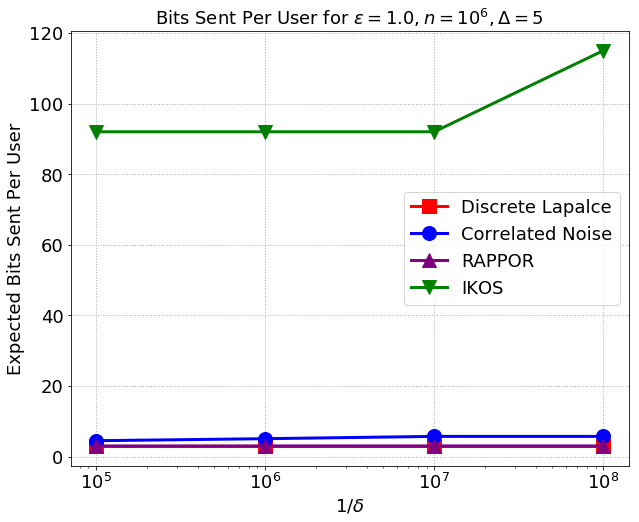}
\caption{Expected bits sent for varying $\delta$}
\end{subfigure}
\begin{subfigure}[b]{0.4\textwidth}
\vspace{10mm}
\includegraphics[width=\textwidth]{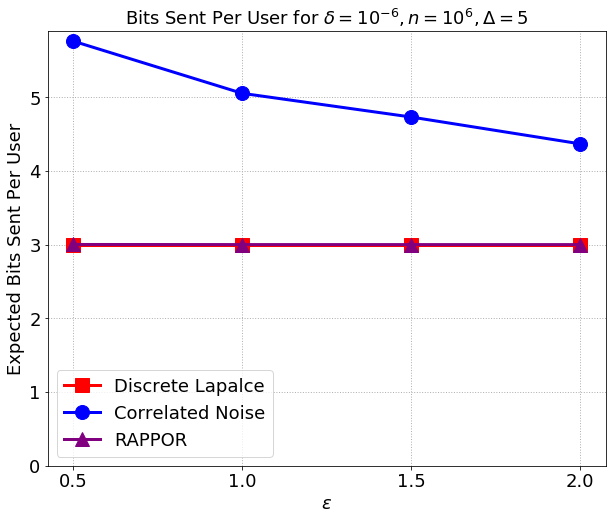}
\caption{Expected bits sent for varying $\eps$ (without IKOS)}
\end{subfigure}
\begin{subfigure}[b]{0.4\textwidth}
\includegraphics[width=\textwidth]{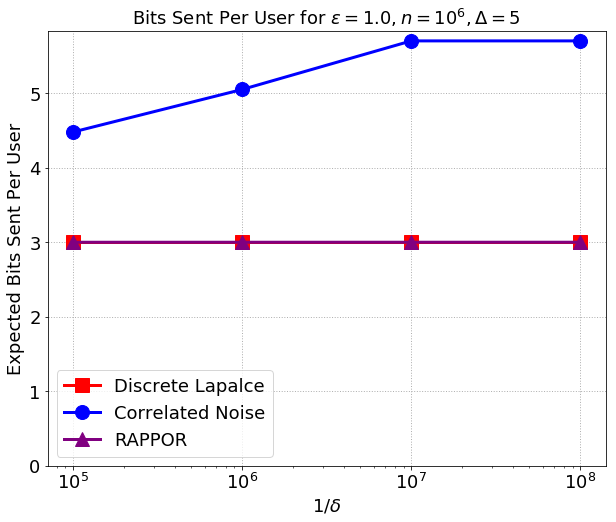}
\caption{Expected bits sent for varying $\delta$ (without IKOS)}
\end{subfigure}
\caption{Communication complexity of our ``correlated noise'' real-summation protocol compared to other protocols on the rent dataset when varying $\eps$ or $\delta$. Since the number of bits sent in the IKOS algorithm is usually very large, it is hard to distinguish the other three lines in the plots where it is included; so we provide the plots without the IKOS algorithm as well.}
\label{fig:comm-varying-eps-delta}
\end{figure*}

\subsection{Experiments on Real Summation}


We evaluate our algorithm on the public $1940$ US Census IPUMS dataset \cite{sobek2010integrated}. We consider the \emph{household rent} feature in this dataset; and consider the task of computing the average rent. We restrict to the rents below $\$ 400$; this reduced the number of reported rents from $68,290,222$ to $66,994,267$, preserving more than $98 \%$ of the data.
We use randomized rounding (see the proof of \Cref{thm:real-main}) with number of discretization levels\footnote{This means that we first use randomized rounding as in the proof of \Cref{thm:real-main}, then run the $\Delta$-summation protocol, and the analyzer just multiplies the estimate by $400/\Delta$.} $\Delta = 20, 50, 100, 200$. Note that 200 is the highest ``natural'' discretization that is non-trivial in our dataset, since the dataset only contains integers between 0 and 400.

We remark that, for all algorithms, the RMSE is input-dependent but there is still an explicit formula for it; the plots shown below use this formula. On the other hand, the communication bounds do not depend on the input dataset. 

\paragraph{Error.}
In terms of the errors, unfortunately all of the shuffle DP mechanisms incur significantly higher errors than the Laplace mechanism in central DP. The reason is that the discretization (i.e., randomized rounding) error is already very large, even for the largest discretization level of $\Delta = 200$. Nonetheless, the errors for both our algorithm and IKOS decrease as the discretization level increases. Specifically, when $\Delta = 200$, the RMSEs for both of these mechanisms are less than half the RMSE of RAPPOR for \emph{any} number of discretization levels. Notice also that the error for RAPPOR increases for large $\Delta$; the reason is that, as shown in \Cref{sec:experiments}, the error of RAPPOR grows with $\Delta$ even without discretization. These results are shown in \Cref{subfig:rent-error}.

\paragraph{Communication.}
In terms of communication, we compare to the \emph{non-private} ``Discretized'' algorithm that just discretizes the input into $\Delta$ levels and sends it to the analyzer; it requires only one message of $\lceil \log \Delta \rceil$ bits. RAPPOR incurs little (less than $1\%$) overhead,  in expected number of bits sent, compared to this baseline. When $\Delta$ is small, our correlated noise algorithm incurs a small amount of overhead (less than $20\%$ for $\Delta = 25$), but this overhead grows moderately as $\Delta$ grows (i.e., to less than $60\%$ for $\Delta = 200$). Finally, we note that, since the number $n$ of users is very large, the message length of IKOS (which must be at least $\log n\Delta$) is already large and thus the number of bits sent per user is very large --- more than 10 times the baseline.

\begin{figure*}
\centering
\begin{subfigure}[b]{0.49\textwidth}
\includegraphics[width=\textwidth]{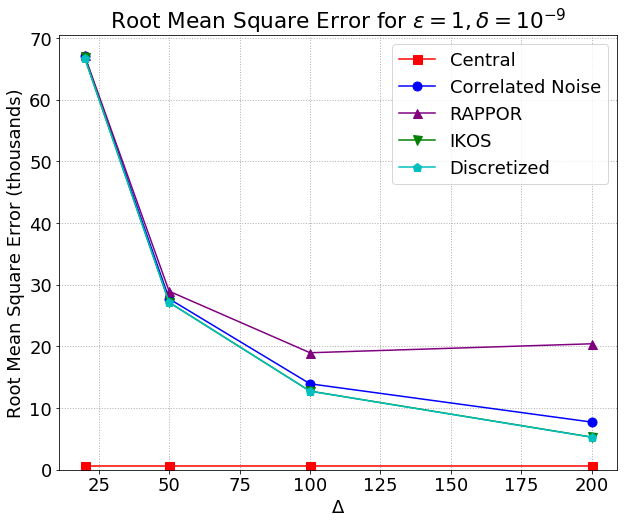}
\caption{RMSEs for varying $\Delta$}
\label{subfig:rent-error}
\end{subfigure}
\begin{subfigure}[b]{0.49\textwidth}
\includegraphics[width=\textwidth]{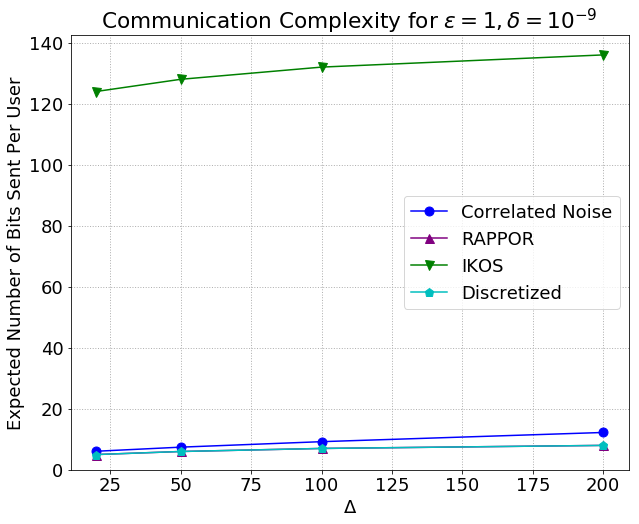}
\caption{Expected bits sent for varying $\Delta$}
\label{subfig:rent-comm}
\end{subfigure}
\caption{Error and communication complexity of our ``correlated noise'' real-summation protocol compared to other protocols on the rent dataset. The Laplace mechanism is included for comparison in the RMSE plot, even though it is not implementable in the shuffle model; it is not included in the communication plot since the communication is not well-defined. For both plots, we also include ``Discretized'' which is a \emph{non-private} algorithm that just discretizes the input and sends it to the analyzer, which then sums up all the incoming messages.}
\end{figure*}

\end{document}